\documentclass[letter,10pt]{article}
\usepackage{fullpage}
\pdfoutput=1

\usepackage{times}
\usepackage[cmex10]{amsmath}
\usepackage{amsfonts}
\usepackage{amsthm}
\usepackage{amsxtra}
\usepackage{amssymb}
\usepackage{textcomp}
\usepackage{wasysym} 
\usepackage{graphicx} 
\usepackage{color} 
\usepackage[noadjust]{cite}

\usepackage[tight]{subfigure}

\usepackage{mathrsfs} 
\usepackage{centernot}
\usepackage{enumerate}
\usepackage{accents} 
\usepackage{ifpdf}

\usepackage{algpseudocode} 
\usepackage{float}
\floatstyle{ruled}
\newfloat{alg}{H}{loalg}
\floatname{alg}{Algorithm}
\floatstyle{plain}

\newtheorem{thm}{Theorem}

\newtheorem{lem}{Lemma}
\newtheorem{prop}{Proposition}

\newtheorem{defn}{Definition}

\DeclareMathOperator{\dist}{dist}

\DeclareMathOperator*{\mini}{min.}

\DeclareMathOperator*{\st}{subj.\ \! to}

\DeclareMathOperator*{\Limsup}{Lim\,sup\,}

\DeclareMathOperator{\ViabOp}{Viab}
\newcommand{\Viab}{\ViabOp^\textup{sd}}

\DeclareMathOperator{\err}{err}
\DeclareMathOperator{\sgn}{sgn}

\DeclareMathOperator{\conv}{conv}
\DeclareMathOperator{\vol}{vol}
\DeclareMathOperator*{\argmax}{\arg\max}

\newcommand{\Y}{\mathcal{Y}}
\newcommand{\V}{\mathcal{V}}

\newcommand{\U}{\mathcal{U}}
\newcommand{\Ul}{\mathscr{U}^\textup{pwc}}

\newcommand{\X}{\mathcal{X}}
\newcommand{\K}{\mathcal{K}}

\newcommand{\A}{\mathcal{A}}
\newcommand{\B}{\mathcal{B}}
\newcommand{\C}{\mathcal{C}}

\newcommand{\T}{\mathbb{T}}
\renewcommand{\S}{\mathcal{S}}
\newcommand{\Lom}{\mathcal{L}}

\newcommand{\tr}[1]{#1^\top} 
\newcommand{\Real}{\mathbb{R}}

\newcommand{\abs}[1]{\left\vert#1\right\vert}

\newcommand{\norm}[1]{\left\lVert#1\right\rVert}
\newcommand{\inner}[1]{\langle#1\rangle}

\usepackage{marginnote} 

\usepackage[bookmarks=true]{hyperref}

\begin{document}

\title{A sampling-based approach to scalable constraint satisfaction in linear sampled-data systems---Part I: Computation\thanks{An earlier version of this paper was partially presented at the 17th Int'l Conference on Hybrid Systems: Computation and Control, April 15-17, 2014, Berlin, Germany \cite{Gillula_HSCC2014}.}} 

\author{Shahab Kaynama, Jeremy H. Gillula, Claire J. Tomlin\thanks{The authors are with Electrical Engineering \& Computer Sciences, University of California at Berkeley, 337 Cory Hall, Berkeley, CA 94720, USA. {\tt\small \{kaynama, jgillula, tomlin\}@eecs.berkeley.edu}}}

\date{(Preprint Submitted for Publication)}

\maketitle


\begin{abstract}                          
    %
    Sampled-data (SD) systems, which are composed of both discrete- and continuous-time components, are arguably one of the most common classes of cyberphysical systems in practice; most modern controllers are implemented on digital platforms while the plant dynamics that are being controlled evolve continuously in time.  As with all cyberphysical systems, ensuring hard constraint satisfaction is key in the safe operation of SD systems.  A powerful analytical tool for guaranteeing such constraint satisfaction is the viability kernel: the set of all initial conditions for which a safety-preserving control law (that is, a control law that satisfies all input and state constraints) exists. In this paper we present a novel sampling-based algorithm that tightly approximates the viability kernel for high-dimensional sampled-data linear time-invariant (LTI) systems.  Unlike prior work in this area, our algorithm formally handles both the discrete and continuous characteristics of SD systems. We prove the correctness and convergence of our approximation technique, provide discussions on heuristic methods to optimally bias the sampling process, and demonstrate the results on a twelve-dimensional flight envelope protection problem.
\end{abstract}

\section{Introduction} \label{S:Intro}

The mathematical guarantee of satisfaction of hard input and state constraints is an increasingly desirable property that every safety-critical cyberphysical system must implement. The subset of the state space for which this property holds is known as the viability kernel~\cite{aubin2011}, or alternatively, in the infinite horizon case, the maximal controlled-invariant set~\cite{blanchini2008set}. Consequently, a tremendous amount of work in the recent literature has been focused on methods for computing the viability kernel.

Constrained sampled-data (SD) systems describe a large class of cyberphysical systems. In most practical settings the system evolves continuously in time but the state is measured only at discrete time instants~\cite{Goodwin_CSSMAG13}. Consequently, any admissible control policy is piecewise constant; the input can only be applied at the beginning of each sampling interval and is kept constant (under zero-order hold) until the next sampling time. Examples of SD systems include control of blood glucose levels in type-1 diabetes ~\cite{diabetes_MPC_2007} where insulin could only be administered every $30$ minutes during the clinical trials. Additionally, restrictions on sampling frequency is not always due to sensory limitations or the particulars of an application. For instance, in model predictive control (MPC) a higher sampling frequency results in a significantly larger online optimization problem over each prediction horizon---a known limiting factor in embedded control design, e.g.\ in the automotive industry. Furthermore, without additional rate constraints, a higher sampling frequency would require controllers with much higher bandwidth.



Discretizing the dynamics and naively designing control policies in discrete time ignores the behavior of the true system in between two sampling instants. This inter-sample behavior can be crucial in safety-critical systems where the state constraint is associated with ``safety'' of the system. At the same time, performing continuous-time safety analysis on the system (for example, using the traditional level-set technique~\cite{MBT05}), which only requires a mild assumption of Lebesgue measurability of the control input, cannot provide guarantees about the behavior of the SD system where the input is restricted to draw from the class of piecewise constant signals~\cite{kaynama_NAHS2012}. All this warrants a technique that can formally handle SD systems.

We present a sampling-based approach that yields tight under- and over-approximations of the viability kernel for \emph{high-dimensional SD} linear time-invariant (LTI) dynamics. 

\subsection{Related Work}

The classical numerical schemes for approximating the viability kernel are those based on gridding the state space (and discretizing the dynamics) and appropriately evolving the constraints over this stationary grid. Such schemes include the level-set methods ~\cite{MBT05} and variants of Saint-Pierre's viability algorithm ~\cite{Saint-Pierre_1994}. Despite their versatility in handling complex dynamics and sets, the applicability of these schemes has historically been limited to systems of dimensions less than five. Efforts to generalize such grid-based techniques to moderately dimensioned systems include structure decomposition~\cite{kaynama_TAC2009, Mitchell_HSCC11} and approximate dynamic programming~\cite{Coquelin_Martin_Munos_2007}.


Algorithms from within the MPC community have also emerged that enable the computation of the viability kernel for discrete-time LTI systems with polytopic constraints~\cite{blanchini2008set}. Due to the fact that these algorithms recursively compute the Minkowski sum, linear transformation, and intersection of polytopes, they can only be applied to low dimensional systems; the number of vertices of the resulting polytope grows rapidly with each subsequent Minkowski sum operation, while the intersection operation at each iteration requires a vertex to facet enumeration of polytopes---an operation that is known to be intractable in high dimensions~\cite{Bremner_1998}. In more general contexts (e.g., for continuous-time systems), an ellipsoidal approximation of the region of attraction of the MPC is computed as a (crude) representation of the viability kernel~\cite{chen2001optimisation}. These, as well as other approximation techniques such as $\beta$-contractive polytopes~\cite{Alessio_Lazar_Bemporad_Heemels_2007}, generally require existence of a stabilizing controller within the constraints---a requirement that may not always be readily satisfied.

For LTI systems with convex constraints, \cite{kaynama_Aut2013} (discrete-time) and \cite{kaynama_HSCC2012} (discrete- and continuous-time) introduced efficient and scalable algorithms based on support vector representations and piecewise ellipsoidal sets to conservatively approximate the viability kernel. These algorithms follow the flow of the dynamics in the same spirit as Lagrangian techniques for maximal reachability such as~\cite{Le_Guernic_Girard_2010, SpaceEx_2011, Kurzhanski_SCL_2000}.

Approximating the viability kernel can also be viewed as a search for an appropriate control Lyapunov function subject to additional input and state constraints. As such, in parallel to the above developments, for polynomial systems with semi-algebraic constraints, various sum-of-squares (SOS)\footnote{SOS is a relaxation of the original semidefinite program.} optimization-based techniques have been proposed~\cite{tedrake_finitetime_13,tedrake_funnels_11, tedrake2010lqr, tan_CLF_04, Prajna_Jadbabaie_2004} that either directly form a polynomial approximation of the viability kernel, or can be modified to do so. The resulting bilinear SOS program is solved either by alternating search (prone to local optima) or through convex relaxations (introduces additional conservatism). The degree of the SOS multipliers, which directly translates to the presumed degree of the polynomial that is to describe the kernel, is commonly kept low (e.g.\ quartic), striking a tradeoff between excessive conservatism and computational complexity. A related SOS-based technique is the recent method of occupation measures~\cite{henrion2012convex, majumdar2013technical} that, while scalable, can only \emph{over}-approximate the desired set (which is insufficient for safety). Though the approximation recovers the true kernel in the limit, the error is not monotonically decreasing with the degree of the multipliers. 

All of the above algorithms are designed for either discrete- or continuous-time dynamics, and very little effort has been dedicated to \emph{sampled-data} system---precisely the systems that, due to their physically descriptive and realistic nature, stand to benefit the most from safety formalism. To the best of our knowledge, the only other scalable work on computing the viability kernel for SD systems is presented in~\cite{kaynama_NAHS2012}, where a piecewise ellipsoidal algorithm is proposed for LTI dynamics with ellipsoidal constraints. Unfortunately due to the projections and cross-products involved in this algorithm, the quality of the approximation degenerates rapidly with the time horizon.

\subsection{Summary of Contributions}

We propose a sampling-based approach that directly handles SD systems, albeit under LTI dynamics and convex constraints. The algorithm provably under-approximates the viability kernel in a scalable and efficient manner. Since the proposed method deals with individual trajectories to infer the evolution of sets of initial conditions, it is not explicitly restricted by a particular shape of the constraints, although its computational complexity does depend on these shapes. 


Our technique yields a ``tight'' approximation in the sense that the resulting set touches the boundary of the true viability kernel from the inside with arbitrary precision up to a numerical constant. The algorithm is sampling based, meaning that the points to be included as boundary points of the viability kernel are sampled according to a probability distribution. We provide a convergence proof that shows that such a sampling-based technique is superior to any alternative deterministic approach.

Section~\ref{S:prob_formulation} formulates the problem we wish to address. Section~\ref{S:methodology} presents our main under-approximation algorithm, while Section~\ref{S:correct_converge} proves its correctness and convergence. Section~\ref{S:complexity_scalability} discusses the computational complexity of the technique and showcases, via an example, its scalability. To provide a measure of conservatism of our technique, we will also present an over-approximation algorithm in Section~\ref{S:over_approx}. This over-approximation algorithm is then utilized in Section~\ref{S:improve_underapprox_bias} to help guide the under-approximation process. This section also describes a few heuristics to bias the random sampling of the state constraint so as to achieve superior convergence and accuracy properties. We study a 12D flight envelope protection problem in Section~\ref{S:examples}, before providing concluding remarks and future directions in Section~\ref{S:conc}.






\section{Problem Formulation} \label{S:prob_formulation}

Consider the LTI system
\begin{equation}\label{E:SD_system}
    \dot{x}(t) =  Ax(t) + Bu(t)
\end{equation}
%
with state $x(t)\in \Real^n$, control input $u(t)\in\U$, where $\U$ is a compact convex subset of $\Real^{m}$. $A$ and $B$ are constant matrices of appropriate dimension.
The state of the system is measured at every time instant $t_k:=k\delta$ for  $k \in \mathbb{Z}_{\geq 0}$ and fixed sampling interval $\delta \in \Real_{>0}$. We are concerned with the evolution of the system over $\T:= [0,\tau]$ with arbitrary, finite time horizon $\tau \in \Real_{>0}$. We denote by $N_\delta := \lceil \tau / \delta \rceil$ the number of sampling intervals in $\T$. The input is applied at the beginning of each sampling interval and is kept constant until the next sampling instant. Thus the input signal draws from the set of piecewise constant functions
\begin{equation}
    \Ul_\T := \{u \colon \T \to \Real^{m} \;\text{piecewise const.}, \; u(t_k)\in \U \; \forall k, \; u(t) = u(t_k)\; \forall t \in [t_k,t_{k+1}) \}.
\end{equation}
For every $x_0\in \Real^n$ and $u(\cdot)\in \Ul_\T$, we denote the (unique) trajectory of \eqref{E:SD_system} by $x_{x_0}^{u}\colon \T \to \Real^n$ with initial condition $x_{x_0}^{u}(0)=x_0$. When clear from the context, we shall drop the subscript and superscript from the trajectory notation.

For a nonempty compact convex state constraint set $\K \subset \Real^n$, deemed \emph{safe}, we examine the following construct:
\begin{defn}[SD Viability Kernel] \label{D:disc}
    The finite-horizon SD viability kernel of $\K$ is the set of all initial states for which there exists a control law such that the trajectories emanating from those states remain in $\K$ over $\T$:
    \begin{equation}\label{E:disc_setbuilder}
        \Viab_\T(\K):=\bigl\{ x_0 \in \K \mid  \exists  u(\cdot)\in \Ul_\T, \;
        \forall  t\in \T, \, x_{x_0}^{u}(t) \in \K \bigr\}.
    \end{equation}
\end{defn}




We seek to find a scalable technique that \emph{under-approximates} the viability kernel. (Any approximation must be conservative, in that at the very least all states for which no admissible input can maintain safety must be excluded.)



\subsection{Preliminaries and Notation}


We say that the vector field in ~\eqref{E:SD_system} is bounded on $\K$ if $\exists M > 0$ such that $\norm{Ax+Bu}_p \leq M$ $\forall (x,u) \in \K \times \U$ for some norm $\norm{\cdot}_p$. If $\K$ and $\U$ are compact, every continuous vector field is bounded on $\K$. The $\norm{\cdot}_p$-distance of a point $x \in \Real^n$ from a nonempty set $\A \subset \Real^n$ is defined as $\dist_p(x,\A) := \inf_{ a \in \A}\norm{x-a}_p$.



%


The \emph{Minkowski sum} of any two nonempty subsets $\A$ and $\C$ is $\A \oplus \C := \{ a+c \mid a \in \A,\, c\in \C \}$; their \emph{Pontryagin difference} (or, the \emph{erosion} of $\A$ by $\C$) is $\A \ominus \C := \{a \mid a \oplus \C \subseteq \A \}$. We denote by $\partial \C$ the boundary of the set $\C$, by $\C^c$ its complement, and by $\vol(\C)$ its volume. $\conv(\{v_0, \ldots, v_N\})$ denotes the convex hull of a set of points $v_0, \ldots, v_N$. $\B_p^n(x,a)$ denotes the closed $p$-norm ball of radius $a \in \Real_{\geq 0}$ about a point $x$ in $\Real^n$, and $\S_p^{n-1}$ the codimension one boundary $\partial \B^n_p(0,1)$. A \emph{ray} in $\Real^n$ is the set of points $\vec{r} = \{r_0 + s r_d \mid s \in \Real_{\ge 0}\}$, where $r_0 \in \Real^n$ is the \emph{origin} of the ray, and $r_d \in \Real^n$ is a unit vector giving the \emph{direction} of the ray.




\section{Methodology}\label{S:methodology}


\subsection{Overview of the Algorithm}

Before we describe our algorithm we first elaborate on some simple subroutines which we make use of, but which we will not formally define due to space constraints.
\begin{itemize}
	\item $\textsc{Find-Intersection-on-Boundary}(\C, \vec{r})$ -- Input: A convex compact set $\C \subset \Real^n$, and a ray $\vec{r}$ with origin $r_0 \in \C$.  Returns a point $x$ along the ray $\vec{r}$ on $\partial \C$.  We note that since $\C$ is convex and compact, and since $r_0 \in \C$, there is exactly one such point $x$.  Runs in time proportional to the number of faces of $\C$ if $\C$ is a polytope, and constant time if $\C$ is an ellipsoid.
	\item $\textsc{Sample-Ray}(x)$ -- Input: A point $x \in \Real^n$.  First samples a point $r_d$ at random from $\S_2^{n-1}$. Returns the ray $\vec{r} = \{x + s r_d\}$. Runs in time linear in $n$.
	\item $\textsc{Feasible}(x, \C)$ -- Input: A point $x \in \Real^n$ and a convex compact set $\C \subset \Real^n$.  Returns true if $x \in \Viab_\T(\C)$, and false otherwise.  Further details on the implementation of this subroutine and its time complexity are given in Sections ~\ref{subsec:Checking_Point_Feasibility} and ~\ref{S:complexity_scalability}.
\end{itemize}



%

\newcommand{\const}{0.7}
\newcommand{\spc}{\hspace{2pt}}
\begin{figure}[t]
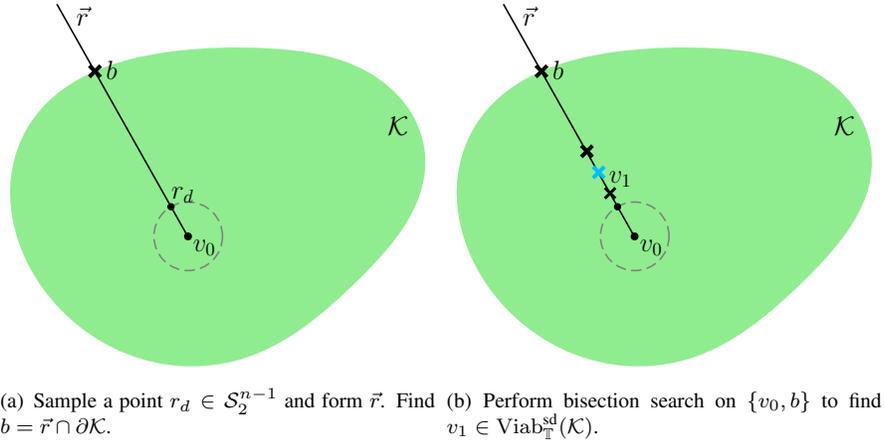
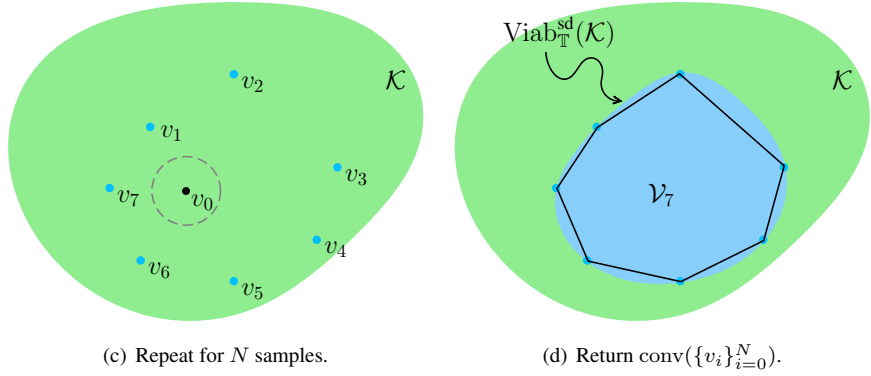

    \centering
    \subfigure[Sample a point $r_d \in \S_2^{n-1}$ and form $\vec{r}$. Find $b = \vec{r}\cap \partial \K$.]{\scalebox{\const}{\input{alg1.TpX}}}
    \spc
    \subfigure[Perform bisection search on $\{v_0,b\}$ to find $v_1 \in  \Viab_\T(\K)$.]{\scalebox{\const}{\input{alg2.TpX}}}\\
    \subfigure[Repeat for $N$ samples.]{\scalebox{\const}{\input{alg3.TpX}}}
    \spc
    \subfigure[Return $\conv(\{v_i\}_{i=0}^N)$.]{\scalebox{\const}{\input{alg4.TpX}}}
    \caption{Illustration of Algorithm~\ref{Alg:polytopic_approx}.}
    \label{fig:iter}
\end{figure}


The under-approximation of the viability kernel proceeds as follows. We assume as input a description of $\K$,\footnote{In theory, $\K$ can be of any arbitrary (but convex) shape. In practice, this shape directly affects the run time complexity of the subroutines \textsc{Find-Intersection-on-Boundary} and \textsc{Feasible} as discussed in the bullet points above.} as well as some initial point $v_0 \in \Viab_\T(\K)$.\footnote{For linear systems if $(0,0) \in \K\times \U$ then $0 \in \Viab_\T(\K)$, and thus we can use $v_0 = 0$.  Otherwise, we assume we can sample points at random over $\K$ until we find a point $x$ via the subroutine $\textsc{Feasible}(x,\K)$ such that $x\in \Viab_\T(\K)$.}  We then construct a polytopic approximation of $\Viab_\T(\K)$ by iteratively sampling a direction $r_d$ and generating a ray $\vec{r}$ centered at $v_0$ (Algorithm~\ref{Alg:polytopic_approx}, step~\ref{line:sample}); finding the point $b \in \Real^n$ where $\vec{r}$ intersects the boundary of $\K$ (Algorithm~\ref{Alg:polytopic_approx}, step~\ref{line:find_boundary_intersection}); and then performing a bisection search along the line segment $\{v_0, b\}$ until we find a point $v_i$ such that $v_i \in \Viab_\T(\K)$ and $\dist_p(v_i, \partial \Viab_\T(\K)) < \epsilon$ for some desired accuracy $\epsilon >0$ in some norm $\norm{\cdot}_p$ (Algorithm~\ref{alg:bisection_feasibility_search}). By repeating for $N$ samples and taking the convex hull of the resulting points $\{v_0, \ldots, v_N\}$ we arrive at a polytope $\V_N \subseteq \Viab_\T(\K)$ which converges to $\Viab_\T(\K)$ (in a manner we shall formalize in Section~\ref{subsec:Algorithm_Convergence}). The conservatism of the algorithm will be formally proven in Section~\ref{subsec:Algorithm_Correctness}. Fig.~\ref{fig:iter} illustrates the under-approximation procedure.

\begin{alg}
    \begin{algorithmic}[1]

    \Procedure{Polytopic-Approx}{$\K, v_0, N$}

    \State $\V_0 \gets \{v_0\}$

	\For{$i = 1$ to $N$} \Comment{$N$ samples}
	\State $\vec{r} \gets \textsc{Sample-Ray}(v_0)$ \label{line:sample}

	\State $b \gets \textsc{Find-Intersection-on-Boundary}(\K, \vec{r})$\label{line:find_boundary_intersection}

	\State $v_{i} \gets \textsc{Bisection-Feasibility}(v_0, b, \K)$
	
    \State $\V_i \gets \conv(\V_{i-1} \cup \{v_i\})$
    \EndFor

    \State \textbf{return} $\V_N$

    \EndProcedure

%
%
%
%

    \end{algorithmic}
    \caption{Computes a polytopic under-approximation of $\Viab_{\T}(\K)$ with at most $N+1$ vertices}
    \label{Alg:polytopic_approx}
\end{alg}

\begin{alg}
    \begin{algorithmic}[1]

    \Function{Bisection-Feasibility}{$a, b, \K$}

	\State $c \gets a+(b-a)/2$
	
	\If{$\textsc{Feasible}(c, \K)$} \label{line:Feasible}
		\If{$\dist_p(b, c) < \epsilon$} \label{line:Epsilon_Threshold}
			\State \textbf{return} $c$
		\Else
			\State \textbf{return} $\textsc{Bisection-Feasibility}(c, b, \K)$
		\EndIf
	\Else
		\State \textbf{return} $\textsc{Bisection-Feasibility}(a, c, \K)$
	\EndIf

    \EndFunction

    \end{algorithmic}
    \caption{Determines an $\epsilon$-accurate intersection of $\partial \Viab_\T(\K)$ and the line between $a$ and $b$.}
    \label{alg:bisection_feasibility_search}
\end{alg}

\subsection{Checking Point Feasibility}
\label{subsec:Checking_Point_Feasibility}

The key to our approach is the subroutine $\textsc{Feasible}(x_0, \K)$ in Step~\ref{line:Feasible} of Algorithm~\ref{alg:bisection_feasibility_search}, which returns true only if $x_0 \in \Viab_\T(\K)$. An overview of the procedure is as follows: Given an initial condition $x_0$, we verify the existence of a piecewise constant control law that ensures that the trajectory starting from $x_0$ belongs to a subset of $\K$ at every sampling instant. This subset is appropriately chosen at a certain distance from the exterior $\K^c$ such that the inter-sampling portions of the trajectory do not escape $\K$. The state $x_0$ is then labeled as \emph{feasible}. We employ forward simulation to determine feasibility. To do this in a tractable fashion, we use a finite-difference approximation of the dynamics. Therefore, we also need to take into account the effect of the discretization error (Fig.~\ref{F:CT_vs_DT} depicts the significance of this error via a trivial example).



\begin{figure}[t]
  \centering
        \includegraphics[clip=true, scale = 0.45]{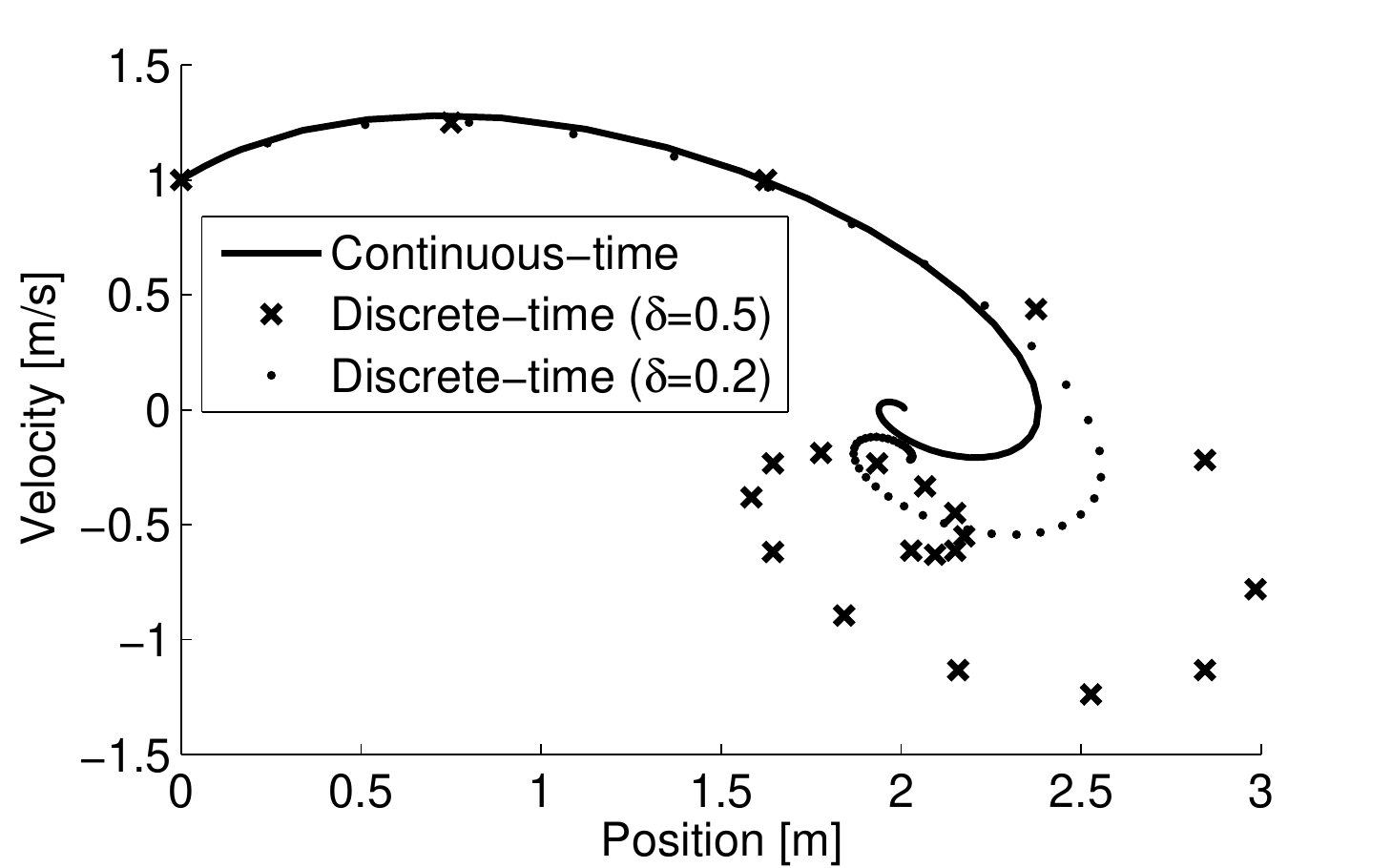}
        \caption{Simulated trajectory of a mass-spring-damper (unit parameters) under constant force in continuous time vs.\ its Euler discretization for two sampling times $\delta = 0.2, 0.5$.} 
        \label{F:CT_vs_DT}
\end{figure}

\subsubsection{Dealing With Inter-Sampling Behavior}

The following lemma ensures that if a control law exists that can keep the trajectory value evaluated at every sampling instant in a certain subset of $\K$, then the continuous evolution of the system in between sampling instants also maintains safety; that is, the curvature of the trajectory during the sampling intervals does not escape $\K$ (Fig.~\ref{fig:erosion1}).



\begin{figure}[t]
    \centering
    \scalebox{0.8}{\input{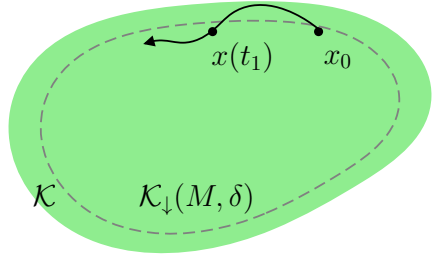}}
    \caption{Erosion of $\K$ ensures that the curvature of the continuous trajectory in between two sampling instants cannot escape safety.}
    \label{fig:erosion1}
\end{figure}

\begin{lem}\label{Lem:HSCC}
    Let~\eqref{E:SD_system} be uniformly bounded on $\K$ by $M>0$ in some norm $\norm{\cdot}_{p_1}$. With a sampling interval $\delta >0$ define
    \begin{equation}\label{E:K_downarrowMd}
    \begin{split}
        \K_\downarrow(M,\delta) &:= \left\{ x \in \K \mid \dist_{p_1}(x,\K^c) \geq M \delta  \right\}\\
        &= \K \ominus \B_{p_1}^n(0,M\delta),
    \end{split}
    \end{equation}
    %
    For a given initial condition $x_0$, if $\exists u(\cdot) \in \Ul_\T$ such that $x(t_k) \in \K_\downarrow(M,\delta)$ $\forall k \in \{0,\dots,N_\delta\}$, then $x_0 \in \Viab_\T(\K)$.
\end{lem}

The proof is similar to the continuous-time analysis in~\cite[Proof of Prop. 1]{kaynama_HSCC2012} and is omitted here; see~\cite[Lem. 1]{Gillula_HSCC2014} for more detail.


\subsubsection{Dealing With Discretization Error}
\label{subsubsec:Dealing_With_Discretization_Error}

Recall that the solution $x_{k+1} = x(t_{k+1})$ of the SD system \eqref{E:SD_system} at time $t_{k+1}$ for any sampling interval $[t_k,t_k+\delta]$ starting from $x_k=x(t_k)$ using a constant input $u_k$ is
\begin{equation}
    x_{k+1} = e^{A\delta} x_k + \biggl(\int_0^\delta e^{A\lambda} d\lambda \cdot B \biggr) u_k. \label{E:SD_solution} 
\end{equation}
%
%
Consider the Taylor series expansion $e^{As} = \sum_{i=0}^\infty (As)^i/i!$. To approximate the evolution of~\eqref{E:SD_system} at every discrete time instant $t_k$ using a finite-difference equation, we approximate the above infinite sum by a $\zeta$th order finite sum
\begin{equation}
    A_{\zeta,s} := \sum\nolimits_{i=0}^\zeta \frac{(As)^i}{i!} \approx e^{As}, \quad  \zeta <\infty
\end{equation}
with truncation error
\begin{equation}
    E_{\zeta,s} := e^{As} - A_{\zeta,s} = \sum\nolimits_{i=\zeta+1}^\infty \frac{(As)^i}{i!}.
\end{equation}
The trajectory of the resulting finite-difference equation is
\begin{equation}
    \hat{x}_{k+1} = A_{\zeta,\delta} \hat{x}_k + \biggl(\int_0^\delta A_{\zeta,\lambda} d\lambda \cdot B \biggr) u_k.
\end{equation}
We refer to $\hat{x}$ as the \emph{nominal} state of the system.

It is important to emphasize that the computation of the matrix exponential has always been a challenging task~\cite{Moler_Van_Loan_2003}. An exact computation is generally not possible. Instead, approximate techniques are employed. For instance, Matlab uses the Pad\'{e} approximation with scaling and squaring method of \cite{Higham2005MatlabExpm}. In most practical cases, the forward Euler approximation ($\zeta =1$) is used.

Truncating the tail $E$ of the Taylor series expansion introduces a discretization error that results in mismatch between the true values of the system trajectory at discrete time instants and the values generated by the finite-difference model. That is, the \emph{discretization error} $e_k$ at time $t_k$ resulting from the truncation error $E$ will cause a deviation from the true state $x_k = x(t_k)$ such that 
\begin{equation}
    \hat{x}_k = x_k - e_k.
\end{equation}
Consequently, to formally guarantee an under-approximation of the SD viability kernel using a time discretized approach---i.e.\ via simulation of the nominal trajectory---we must take into account the effect of the error $e$ and its forward propagation in time (Fig.~\ref{fig:disc_error}).

\begin{figure}[t]
    \centering
    \scalebox{0.9}{\input{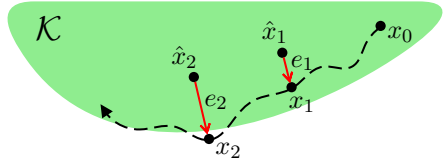}}
    \caption{The mismatch between the true state $x_k$ at sampling time $t_k$ and the nominal state $\hat{x}_k$ of the discretized model. The error $e_k$ propagates in time. If not accounted for in forward simulations, this mismatch could jeopardize safety.}
    \label{fig:disc_error}
\end{figure}

\begin{lem}\label{Lem:gamma_k}
    Suppose that there exist constants $\gamma_k \geq 0$ and a norm $\norm{\cdot}_{p_2}$ such that $\norm{e_k}_{p_2} \leq \gamma_k$ $\forall k$. Let
    \begin{equation}
        \K^k_\downarrow(M,\delta,\gamma_k) := \K_\downarrow(M,\delta) \ominus \B_{p_2}^n(0,\gamma_k) \neq \emptyset.
    \end{equation}
    Then we have that
    \begin{equation}
        \hat{x}_k \in \K^k_\downarrow(M,\delta,\gamma_k) \Rightarrow x_k \in \K_\downarrow(M,\delta).
    \end{equation}
\end{lem}

\begin{proof}
    Regardless of the perturbation caused by the error, since $e_k \in \B_{p_2}^n(0,\gamma_k)$ we have that $x_k \in \{\hat{x}_k\} \oplus \B_{p_2}^n(0,\gamma_k)$. 
    By enforcing the condition $\hat{x}_k \in \K^k_\downarrow(M,\delta,\gamma_k)$ we guarantee that $x_k \in \left( \K_\downarrow(M,\delta) \ominus \B_{p_2}^n(0,\gamma_k) \right) \oplus \B_{p_2}^n(0,\gamma_k) \subseteq \K_\downarrow(M,\delta)$ since for any nonempty sets $\X$ and $\Y$ it holds true that $(\X \ominus \Y) \oplus \Y \subseteq \X$.
\end{proof}


For any given initial condition $x_0$ we trivially have $e_0 = 0$ and thus $\K^0_\downarrow(M,\delta,\gamma_0) = \K_\downarrow(M,\delta)$. We now describe a procedure to compute error bounds $\gamma_k$ which will be used for \emph{a priori} construction of the sets $\K^k_\downarrow(M,\delta,\gamma_k)$.

%
Using the identity $e^{As} = A_{\zeta,s}+E_{\zeta,s}$ in \eqref{E:SD_solution} yields
\begin{align}
    x_k 
    &= \underbrace{A_{\zeta,\delta} x_{k-1} + \biggl(\int_0^\delta A_{\zeta,\lambda} d\lambda \cdot B \biggr) u_{k-1}}_{\hat{x}_k} \notag \\
     &\qquad + \underbrace{E_{\zeta,\delta} x_{k-1} + \biggl(\int_0^\delta E_{\zeta,\lambda} d\lambda \cdot B \biggr) u_{k-1}}_{e_k}.
\end{align}
Back-substituting the solutions $x_{k-1}$ into $x_k$, and $\hat{x}_{k-1}$ into $\hat{x}_k$ for every $k=1,\dots,N_\delta$ we can rewrite $e_k = x_k - \hat{x}_k$ as
\begin{multline}\label{E:e_k_generic}
    e_k = \left( (A_{\zeta,\delta} + E_{\zeta,\delta})^k - A_{\zeta,\delta}^k \right) x_0
    + \sum_{i=0}^{k-1} \biggl[ \left((A_{\zeta,\delta} + E_{\zeta,\delta})^i - A_{\zeta,\delta}^i \right)\\
    \times \int_0^\delta A_{\zeta,\lambda} d\lambda +  (A_{\zeta,\delta} + E_{\zeta,\delta})^i \int_0^\delta E_{\zeta,\lambda} d\lambda  \biggr] B u_{k-1-i}.
\end{multline}
We can do so because for any discrete-time LTI system $x_{k+1} = \Phi x_k + \Psi u_k$, the solution $x_k$ can be written in terms of the initial condition $x_0$ and the past inputs as $x_k = \Phi^k x_0 + \sum_{i=0}^{k-1} \Phi^i \Psi u_{k-1-i}$.


To compute the upper-bound $\gamma_k$ on \eqref{E:e_k_generic}, invoke i) the inequality $\norm{A^k} \leq \norm{A}^k$ which holds for any matrix $A$ and positive constant $k$, ii) the multiplicative and triangular inequalities, and iii) the binomial expansion for $(A_{\zeta,\delta} + E_{\zeta,\delta})^k$ (which is valid since $A_{\zeta,s} E_{\zeta,s} = E_{\zeta,s} A_{\zeta,s}$ \cite{Standish1975179}):
\begin{align}
    \norm{e_k}
    %
    &\leq \norm{ \sum_{l=0}^k \binom{k}{l} A_{\zeta,\delta}^l E_{\zeta,\delta}^{k-l} - A_{\zeta,\delta}^k} \norm{x_0} \notag\\
    & + \sum_{i=0}^{k-1} \left[  \norm{ \sum_{l=0}^i \binom{i}{l} A_{\zeta,\delta}^l E_{\zeta,\delta}^{i-l} - A_{\zeta,\delta}^i } \norm{ \int_0^\delta A_{\zeta,\lambda} d\lambda } \right. \notag\\
    &  \left. +  \norm{ A_{\zeta,\delta} + E_{\zeta,\delta} }^i \norm{ \int_0^\delta E_{\zeta,\lambda} d\lambda } \right]  \sup_{u\in \U}\norm{Bu}.\label{E:ek_inequality_partial}
\end{align}
For $k=1$ this inequality is
\begin{equation}\label{E:e_1_bound}
    \norm{e_1} \leq \norm{ E_{\zeta,\delta} } \norm{x_0} + \norm{ \int_0^\delta E_{\zeta,\lambda} d\lambda } \sup_{u\in \U}\norm{Bu}.
\end{equation}
For $k>1$ we find
\begin{align}
    \norm{e_k}
    %
    &\leq \sum_{l=0}^{k-1} \binom{k}{l} \norm{A_{\zeta,\delta}^l}  \norm{E_{\zeta,\delta}}^{k-l} \norm{x_0} \notag\\
    & + \sum_{i=1}^{k-1} \left[  \sum_{l=0}^{i-1} \binom{i}{l} \norm{A_{\zeta,\delta}^l} \norm{E_{\zeta,\delta}}^{i-l} \norm{ \int_0^\delta A_{\zeta,\lambda} d\lambda } \right. \notag\\
    &  \left. +  \left(\norm{A_{\zeta,\delta}} + \norm{E_{\zeta,\delta}}\right)^i \norm{ \int_0^\delta E_{\zeta,\lambda} d\lambda } \right]  \sup_{u\in \U}\norm{Bu}. \label{E:e_k_bound}
\end{align}

The term $\int_0^\delta A_{\zeta,\lambda} d\lambda$ in \eqref{E:e_k_bound} is a definite integral of a finite sum and it can be easily computed:
%
\begin{equation}\label{E:integral_M}
    \int_0^\delta A_{\zeta,\lambda} d\lambda = \sum_{i=0}^\zeta \int_0^\delta \frac{(A \lambda)^i}{i!} d\lambda  = \sum_{i=0}^\zeta \frac{A^i \delta^{i+1}}{(i+1)!}.
\end{equation}
Evaluating the integral $\int_0^\delta E_{\zeta,\lambda} d\lambda$ is trickier. We can, however, compute an upper-bound on its $\infty$-norm. We will use the following property \cite{Liou_1966}:
%
\begin{align}
    \norm{E_{\zeta,\delta}}_\infty
    &\leq \frac{(\norm{A}_\infty \delta)^{\zeta+1}}{(\zeta + 1)!} \cdot \frac{1}{1-\varepsilon} =: \psi_\delta, \label{E:Es_upperbound}
\end{align}
where the discretization order $\zeta$ is chosen such that $\varepsilon := \frac{\norm{A}_\infty \delta}{\zeta+2} < 1$ (to ensure convergence of the power series $1 + \varepsilon +\varepsilon^2 + \cdots$). 
%
%
%
Similarly, we can derive
%
%
\begin{align}
    \norm{\int_0^\delta E_{\zeta,\lambda} d\lambda}_\infty &\leq \int_0^\delta \norm{E_{\zeta,\lambda}}_\infty d\lambda
    \leq \int_0^\delta \sum_{i = \zeta +1}^\infty \frac{\norm{A}_\infty^i \lambda^i}{i!} d\lambda \notag\\
    = &\sum_{i = \zeta +1}^\infty \int_0^\delta \frac{\norm{A}_\infty^i \lambda^i}{i!} d\lambda
    = \sum_{i = \zeta +1}^\infty \frac{\norm{A}_\infty^i \delta^{i+1}}{(i+1)!} \notag \\
    \leq &\frac{\norm{A}_\infty^{\zeta+1}\delta^{\zeta+2}}{\left( \zeta+2 \right)!} \cdot \frac{1}{1-\eta} \leq \psi_\delta  \cdot \frac{\delta}{\zeta+2}, \label{E:Es_integral_upperbound}
\end{align}
%
%
where $\eta = \varepsilon \cdot \bigl(1- \frac{1}{\zeta+3}\bigr)$. Note that $\varepsilon <1 \Rightarrow \eta <1$, which automatically ensures the convergence of the power series $1+\eta + \eta^2 + \cdots$ used in derivation of \eqref{E:Es_integral_upperbound}.

Substituting \eqref{E:integral_M}--\eqref{E:Es_integral_upperbound} into \eqref{E:e_1_bound}--\eqref{E:e_k_bound} for the $\infty$-norm we obtain a conservative bound $\tilde{\gamma}_k$ on $\norm{e_k}_\infty$ as
\begin{align}
    \tilde{\gamma}_1 &:= \psi_\delta \norm{x_0}_\infty + \left( \psi_\delta  \cdot \tfrac{\delta}{\zeta+2} \right)  \sup_{u\in \U}\norm{Bu}_\infty; \label{E:e_1_bound_star}\\
    \tilde{\gamma}_k &:= \sum_{l=0}^{k-1} \binom{k}{l} \norm{A_{\zeta,\delta}^l}_\infty  \psi_\delta^{k-l} \norm{x_0}_\infty \notag\\
    &\quad + \sum_{i=1}^{k-1} \Biggl[  \sum_{l=0}^{i-1} \binom{i}{l} \norm{A_{\zeta,\delta}^l}_\infty \psi_\delta^{i-l} \biggl\Vert \sum_{j=0}^\zeta \frac{A^j \delta^{j+1}}{(j+1)!} \biggr\Vert_\infty \notag\\
    &\quad +  \left(\norm{A_{\zeta,\delta}}_\infty + \psi_\delta \right)^i \psi_\delta  \cdot \tfrac{\delta}{\zeta+2} \Biggr]  \sup_{u\in \U}\norm{Bu}_\infty \label{E:e_k_bound_star}
\end{align}
for $k>1$. Using the upper-bounds \eqref{E:e_1_bound_star}--\eqref{E:e_k_bound_star} in Lemma~\ref{Lem:gamma_k} allows us to check for feasibility of a given point $x_0$ via only the finite-difference model, while ensuring that safety will not be violated due to discretization.


The bound $\tilde{\gamma}_k$ is asymptotically tight in the sense that for any $k$, $\lim_{\zeta \to \infty} \tilde{\gamma}_k = 0$. In practice, the chosen order of discretization $\zeta$ must be large enough so as to ensure convergence of the power series in derivation of \eqref{E:Es_upperbound} as well as non-emptiness of the eroded sets in Lemma~\ref{Lem:gamma_k}.


\subsubsection{Verifying Feasibility of $x_0$ via Forward Simulation}

We can now simply use the discretized model
\begin{equation}\label{E:DT_system}
    \hat{x}_{k+1} = A_{\zeta,\delta} \hat{x}_k + B_{\zeta,\delta} u_k
\end{equation}
with $B_{\zeta,\delta} := \int_0^\delta A_{\zeta,\lambda} d\lambda B$ to determine feasibility of a given initial condition $x_0$ without worrying about the discretization error or the inter-sampling behavior of the continuous trajectories of \eqref{E:SD_system} and their potentially negative impact on safety: If there exists a sequence of controls $\{u_k\}$ so that the nominal states $\hat{x}_k$ of the closed-loop system belong to the precomputed sets $\K^k_\downarrow(M,\delta,\tilde{\gamma}_k)$ as described above, then via Lemmas~\ref{Lem:gamma_k} and \ref{Lem:HSCC} the trajectories of~\eqref{E:SD_system} never exit $\K$.


Let us now construct the prediction equation as in ~\eqref{eqn:Prediction} (where we have used the notation $G x_0 + H \mathbf{u}$ to abbreviate the right-hand side of the equality), and formulate the finite horizon feasibility program
\begin{figure*}[!t]
    \begin{equation}\label{eqn:Prediction}
        \begin{bmatrix}
            \hat{x}_0\\ \hat{x}_1 \\ \hat{x}_2 \\ \vdots \\ \hat{x}_{N_\delta}
        \end{bmatrix}
        =
        \underbrace{\begin{bmatrix}
            I\\ A_{\zeta,\delta} \\ A_{\zeta,\delta}^2 \\ \vdots \\ A_{\zeta,\delta}^{N_\delta}
        \end{bmatrix}}_{G} x_0
        +
        \underbrace{\begin{bmatrix}
            0 & 0 & \dots & 0\\
            B_{\zeta,\delta} & 0 & \dots & 0\\
            A_{\zeta,\delta} B_{\zeta,\delta}  & B_{\zeta,\delta} & \dots & 0\\
            \vdots & \vdots & \dots & \vdots\\
            A_{\zeta,\delta}^{N_\delta-1} B_{\zeta,\delta}  & A_{\zeta,\delta}^{N_\delta-2} B_{\zeta,\delta} & \dots & B_{\zeta,\delta}\\
        \end{bmatrix}}_{H}
    \underbrace{\begin{bmatrix}
            u_0\\
            u_1\\
            \vdots\\
            u_{N_\delta-1}
        \end{bmatrix}}_{\mathbf{u}}
    \end{equation}
    \hrulefill
    \vspace*{4pt}
\end{figure*}
\begin{subequations}\label{E:optimization_general}
    \begin{align}
        \mini_{\mathbf{u}} \quad &0 \\
        \st \quad & \mathbf{u} \in \U^{N_\delta}\label{E:opt_control_constraint}\\
        &\hat{x}_k \in \K^k_\downarrow(M,\delta,\tilde{\gamma}_k), \quad k = 0,\dots,N_\delta \label{E:opt_state_constraint}\\
        &\left[\hat{x}_0 \; \cdots \; \hat{x}_{N_\delta}\right]^\top
        = G x_0 + H \mathbf{u}. \label{E:prediction_constraint}
    \end{align}
\end{subequations}



\begin{thm}\label{Thm:if_u_then_feasible}
    If $\exists \mathbf{u}^*$ satisfying \eqref{E:optimization_general}, then $x_0 \in \Viab_\T(\K)$.
\end{thm}

\begin{proof}
    The proof follows directly from Lemmas~\ref{Lem:HSCC} and \ref{Lem:gamma_k}. More specifically, the prediction equation \eqref{E:prediction_constraint}, for a fixed input sequence $\mathbf{u}$, generates a forward simulation of the finite-difference model \eqref{E:DT_system} over the desired horizon $[0,N_\delta] \cap \mathbb{Z}$ corresponding to the continuous time horizon $[0,N_\delta \delta] = \T$. Constraint \eqref{E:opt_control_constraint} ensures that this input sequence is point-wise admissible (meaning that every member of the sequence belongs to $\U$), while constraint \eqref{E:opt_state_constraint} restricts $\hat{x}_k$ so that the trajectory of \eqref{E:SD_system} evaluated at $t_k$ belongs to $\K_\downarrow(M,\delta)$ since via Lemma~\ref{Lem:gamma_k} $\hat{x}_k \in \K^k_\downarrow(M,\delta,\tilde{\gamma}_k) \Rightarrow x_k = x(t_k) \in \K_\downarrow(M,\delta)$. Lemma~\ref{Lem:HSCC} then automatically guarantees that $x(t)\in \K$ $\forall t\in \T$ which implies $x_0\in \Viab_\T(\K)$.
\end{proof}

The subroutine \textsc{Feasible} in Algorithm~\ref{alg:bisection_feasibility_search} employs Theorem~\ref{Thm:if_u_then_feasible} to determine the feasibility of a given sample point $x_0$. Its computational complexity is proportional to the complexity of~\eqref{E:optimization_general} which, with polytopic constraints for example, is simply a linear program (LP).\footnote{We note that the vast majority of the calculations for $\tilde{\gamma}_k$ from~\eqref{E:e_1_bound_star}--\eqref{E:e_k_bound_star} can be done once and ahead of time. Online, to be able to construct $\K_\downarrow^k$ in~\eqref{E:opt_state_constraint}, two simple operations (a multiplication by $\norm{x_0}$ and an addition) are all that is needed to form $\tilde{\gamma}_k$.}

\section{Conservatism and Convergence}\label{S:correct_converge}


\subsection{Algorithm Correctness}\label{subsec:Algorithm_Correctness}

\begin{thm}\label{Thm:algorithm_correctness}
Given convex sets $\K$ and $\U$ and an initial point $v_0 \in \Viab_\T(\K)$, $\V_N = \textsc{Polytopic-Approx}(\K, v_0, N)$ is a subset of $\Viab_\T(\K)$ $\forall N$.
\end{thm}

\begin{proof}
(By induction)  First, it is obvious that for $N=0$, $\V_0 = \textsc{Polytopic-Approx}(\K, v_0, 0)=\{v_0\} \subseteq \Viab_\T(\K)$, since we are given that $v_0 \in \Viab_\T(\K)$.

Next, assume that $\V_{N-1} = \textsc{Polytopic-Approx}(\K, v_0, N-1)=\conv(\{v_0, \ldots, v_{N-1}\}) \subseteq \Viab_\T(\K)$, and that we have $v_N = \textsc{Bisection-Feasibility}(v_0, b, \K)$ for some point $b \in \partial \K$.  Let $\V_N = \conv(\V_{N-1} \cup \{v_N \})$.  Since $\textsc{Bisection-Feasibility}$ only returns points which are inside $\Viab_\T(\K)$, we know $v_N \in \Viab_\T(\K)$.  Now since $\V_N$ is convex, $\forall x_0 \in \V_N$ $\exists x'_0 \in \V_{N-1}$ and $\exists \theta \in [0,1]$ s.t.\ $x_0 = \theta x'_0 + (1-\theta) v_N$. For $x'_0$ we know (by induction hypothesis) $\exists u_{x'_0}(\cdot)\in \Ul_{\T}$ s.t.\ $x'(t) = e^{At}x'_0 + \int_0^t e^{A(t-r)} B u_{x'_0}(r) dr \in \K$ $\forall t \in \T$. For $v_N$ we also know $\exists u_{v_N}(\cdot) \in \Ul_\T$ s.t.\ $x''(t) = e^{At} v_N + \int_0^t e^{A(t-r)} B u_{v_N}(r) dr \in \K \quad \forall t \in \T$. Therefore,
    \begin{align}
        \tilde{x}(t)
        &:= \theta x'(t) + (1-\theta) x''(t)= e^{At} \left( \theta x'_0 + (1-\theta) v_N \right) \notag\\
         &\qquad + \int_0^t e^{A(t-r)} B \left( \theta u_{x'_0}(r) + (1-\theta) u_{v_N}(r) \right) dr\notag \\
        &= e^{At} x_0 + \int_0^t e^{A(t-r)} B u_{x_0}(r) dr \in \K \;\; \forall t \in \T \label{E:safety_preserving_combo}
    \end{align}
    since $\K$ and $\U$ are convex and compact. Thus, $u_{x_0}(\cdot) = \theta u_{x'_0}(\cdot) + (1-\theta) u_{v_N}(\cdot) \in \Ul_{\T}$ is safety-preserving and $x_0 \in \Viab_\T(\K)$. Because $x_0$ was chosen arbitrarily in $\V_N$, we conclude that $\V_N \subseteq \Viab_\T(\K)$.
\end{proof}

\subsection{Algorithm Convergence} \label{subsec:Algorithm_Convergence}

One of the striking features of our algorithm is that it is random; additional points on the boundary of the polytopic approximation of $\Viab_\T(\K)$ are iteratively generated based on a random sampling.  This random nature is due to the fact that $\partial \Viab_\T(\K)$ is unknown \emph{a priori} (and is, in fact, what we are trying to estimate) so it is impossible to know what points to sample to construct a polytope that converges to $\Viab_\T(\K)$ as quickly as possible.  In fact, any algorithm which deterministically chose points for which to verify feasibility could be presented with a safe set $\K$ and system dynamics for which the algorithm would converge arbitrarily poorly.  This fact is related to results in the literature of estimating the volume of convex bodies using a \emph{separation oracle}\footnote{A separation oracle is a function that accepts as input a convex set and a point, and returns whether or not that point is inside the convex set.  In our algorithm $\textsc{Feasible}(x, \K)$ plays this role.}. More specifically, part of the literature on algorithms for estimating the volume of convex bodies states that it can be shown that for any algorithm that deterministically queries a separation oracle a polynomial number of times to build a polytopic approximation, the error (the difference in volume between the approximation and the true set) could be exponential in the number of dimensions~\cite{Barany1987}.  Random algorithms, on the other hand, can perform in a provably better manner~\cite{Dyer1991}.

To prove our algorithm's asymptotic convergence, first let $\overline{\Viab_\T}(\K)$ be the subset of the viability kernel we are actually attempting to approximate, i.e.
\begin{equation}
    \overline{\Viab_\T}(\K) := \{x_0 \mid \exists  u(\cdot)\in \Ul_\T,\, \forall k, \, x(t_k) \in \K_\downarrow(M, \delta) \}.
\end{equation}
While the true kernel contains all initial conditions for which a piecewise constant control keeps $x(t) \in \K$, the above set only encompasses initial conditions for which $x(t_k) \in \K_\downarrow(M, \delta)$ (which is a sufficient, but not necessary, condition to imply $x(t)\in \K$; cf.\ Lemma~\ref{Lem:HSCC}). Define the volumetric error between these two sets as
\begin{equation}\label{eqn:Definition_of_Epsilon_Cont}
	\epsilon_{\mathrm{cont}}(M\delta) := \vol(\Viab_\T(\K)) - \vol(\overline{\Viab_\T}(\K))
\end{equation}
and note that it depends only on the term $M\delta$, due to the definition of $\K_\downarrow(M, \delta)$ in \eqref{E:K_downarrowMd}.

Next, consider the output $\V_N$ of Algorithm~\ref{Alg:polytopic_approx}. Clearly, the accuracy of under-approximation of the viability kernel by the set $\V_N$ is implicitly dependant on the discretization order $\zeta$ (Section~\ref{subsubsec:Dealing_With_Discretization_Error}), and the accuracy $\epsilon$ of the bisection search (Algorithm~\ref{alg:bisection_feasibility_search}). To reflect this dependency, we adapt the extended notation $\V_N^{\zeta, \epsilon}$.\footnote{Theorem~\ref{Thm:algorithm_correctness}, restated in terms of the extended notation, asserts that $\V_N^{\zeta, \epsilon} \subseteq \Viab_\T(\K)$ $\forall \zeta, \epsilon, N$.} Evidently, for fixed values of $\zeta$ and $\epsilon$ as $N \to \infty$, this set only approximates a subset $\overline{\Viab_\T}(\K,\zeta,\epsilon)$ of the set $\overline{\Viab_\T}(\K)$:
\begin{align}
	\Limsup_{N \to \infty} \V_N^{\zeta, \epsilon} &= \overline{\Viab_\T}(\K,\zeta,\epsilon)\\
    \Limsup_{\zeta \to \infty, \epsilon \to 0} \overline{\Viab_\T}(\K,\zeta,\epsilon) &= \overline{\Viab_\T}(\K)
\end{align}
with $\Limsup$ denoting the Kuratowski upper-limit. We are now ready to present our algorithm's convergence property.

\begin{prop}[Rate of Convergence]\label{prop:Convergence}
    Let $\epsilon_{\mathrm{vol}}(N, \zeta, \epsilon)$ be the volumetric error between the viability kernel and the output of our algorithm, minus the error $\epsilon_{\mathrm{cont}}(M\delta)$ between the true kernel and $\overline{\Viab_\T}(\K)$:
    %
    \begin{equation*}
    	\epsilon_{\mathrm{vol}}(N, \zeta, \epsilon) := \vol(\Viab_\T(\K))
    	- \vol(\V_N^{\zeta, \epsilon}) - \epsilon_{\mathrm{cont}}(M\delta).
    \end{equation*}
    Then our algorithm converges as
    \begin{equation}\label{E:convergence_rate}
        \lim_{\substack{
        N \to \infty \\
        \zeta \to \infty \\
        \epsilon \to 0}} \epsilon_{\mathrm{vol}}(N, \zeta, \epsilon) N^\frac{2}{n-1} = c_n(\Viab_\T(\K), M\delta),
    \end{equation}
    where $c_n(\Viab_\T(\K), M\delta)$ is a constant dependent on the dimension $n$, the shape of the viability kernel (more specifically its Gauss-Kronecker curvature), and the value $M\delta$.
\end{prop}

The proof requires some background on random algorithms for convex bodies and is provided in Appendix~\ref{S:appendix_convergence}.

Proposition~\ref{prop:Convergence} asserts that, for fixed dimension $n$, the volumetric error between the outcome of our algorithm and the true viability kernel asymptotically converges, at the exponential rate of $\hat{c}_n(\Viab_\T(\K))/N^\frac{2}{n-1}$, to a numerical constant due to the sampled-data nature of the system. On the other hand, to keep the accuracy of the approximation the same as $n \to kn$ we would need an increase of $N \to N^{\frac{kn-1}{n-1}}$. However, the fact that we only store samples on the boundary of the viability kernel to describe that set (as opposed to storing a grid of the entire set $\K$ and possibly beyond) requires significantly less memory than conventional approaches such as the SD level-set method in ~\cite{kaynama_NAHS2012}. The flexibility in choosing the number of samples strikes a direct tradeoff between accuracy and computational complexity, making our algorithm scalable to high dimensions. The computed approximation is far more accurate (and quite possibly more scalable) than the piecewise ellipsoidal technique also presented in \cite{kaynama_NAHS2012}.


Again, due to the results in~\cite{Schutt2003}, the above convergence rate is optimal (up to a multiplicative constant depending on the probability density function used for sampling); no other algorithm that approximates the kernel by sampling from its boundary will be able to converge at a faster rate.

\section{Computational Complexity \& Scalability}\label{S:complexity_scalability}

The run time complexity of our algorithm (for fixed number of sampling intervals $N_\delta$) is $\mathcal{O}(N \log(d) \Phi(n))$, where $N$ is the number of samples/vertices, $d$ is the ``diameter'' of the set $\K$, and $\Phi$ is the complexity of the feasibility program~\eqref{E:optimization_general} as a function of the state dimension $n$. That is, the algorithm runs in time linear in the number of samples $N$, logarithmic in the diameter $d$ of $\K$ due to complexity of the bisection search, and proportional to $\Phi$ in the complexity of the appropriate feasibility program~\eqref{E:optimization_general}. For instance, with polytopic constraints, the feasibility problem is an LP and thus the algorithm runs in time sub-cubic in $n$. This is a direct improvement over existing techniques for approximating the SD viability kernel. Furthermore, since each vertex is processed completely independently of others, our algorithms is highly parallelizable.

To demonstrate the scalability of our algorithm, consider the chain of $n$ integrators $d^n x / dt^n = u$ with constraints $\U = [-0.15, 0.15]$ and $\K = \{x \mid \norm{x}_\infty \leq 0.5 \}$. The state is measured every $\delta = 0.05\,\text{s}$ and safety is to be maintained over $\T = [0,1]$. We use a discretization order of $\zeta = 4$, bisection accuracy of $\epsilon = 0.01$ with maximum of three-level bisection depth, and employ YALMIP \cite{Yalmip_Lofberg04} to implement \eqref{E:optimization_general} and MPT \cite{KGBM04} for simple operations with polytopes. All of these parameters are kept constant as we increase the dimension $n$ and the number of samples $N= 2n$ to examine the scalability of our algorithm. The results are shown in Fig.~\ref{F:scalability_sampling_based}. The algorithm is implemented in MATLAB R2011b and tested on an Intel Core~i7 at $2.9\,\text{GHz}$ with $16\, \text{GB}$ RAM running $64$-bit Windows~7 Pro (without optimizing the code for speed).

\begin{figure}[t]
  \centering
        \includegraphics[clip=true, scale=0.55]{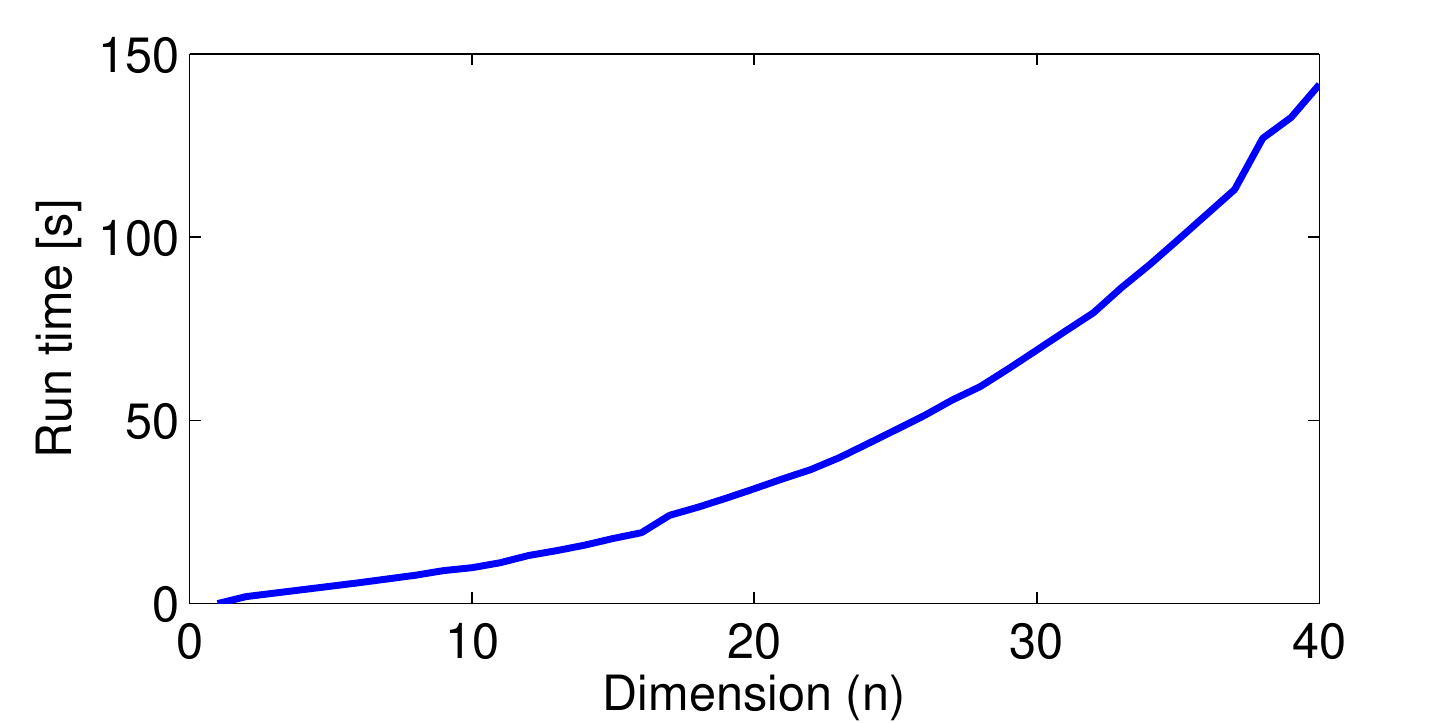}
        \caption{Run time of the algorithm for a chain of $n$ integrators.}
        \label{F:scalability_sampling_based}
\end{figure}




\section{Bounding the Error for Finite Number of Samples: Computing a Tight Over-Approximation}\label{S:over_approx}

Every convex set can be over-approximated by any finite collection of its \emph{support functions}. The support function of a convex compact set $\C \subset \Real^n$ along $\ell \in \Real^n$ is
\begin{equation}
    \rho_\C(\ell) := \max_{x\in \C} \tr{\ell} x.
\end{equation}
The half-space $\{x \mid \tr{\ell} x \leq \rho_\C(\ell)\}$ contains $\C$, and the hyper-plane $\{x \mid \tr{\ell} x = \rho_\C(\ell)\}$ is a supporting hyperplane for $\C$ with normal vector $\ell$ and distance value $\rho_\C(\ell)$. It follows that $\C \subseteq  \bigcap_{\ell \in \Lom} \{x \mid \tr{\ell} x \leq \rho_\C(\ell)\}$ with $\Lom$ a finite subset of $\Real^n$.
%

Let $r_{d_i}$ be the direction $r_d$ along which we have determined the vertex $v_i$ of the under-approximation set through the $i$th iteration of Algorithm~\ref{Alg:polytopic_approx}. It is easy to compute the support function of the set $\K$ along this direction: $\rho_\K(r_{d_i}) = \max_{x\in\K} \tr{r_{d_i}} x$. To find the supporting hyperplane of the true, unknown viability kernel (or rather some appropriate approximation of it) in the direction $r_{d_i}$, we move the hyperplane
\begin{equation}\label{E:support_function_plane}
    \{ x \mid \tr{r_{d_i}} x = \rho_\K(r_{d_i})\}
\end{equation}
on the interior of $\K$ until we find at least one feasible point on this plane that belongs to the kernel (or its over-approximation); see Fig.~\ref{fig:over_approx}.

\begin{figure}[t]
    \centering
    \scalebox{0.9}{\input{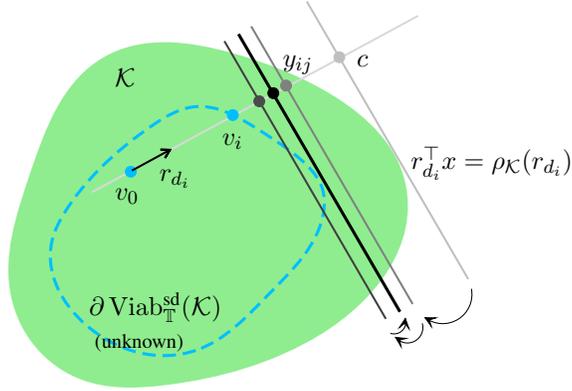}}
    \caption{An $\epsilon_\mathrm{o}$-accurate bisection search between $v_i$ and $c$ determines via \eqref{E:optimization_overapprox} the supporting hyperplane of the viability kernel.}
    \label{fig:over_approx}
\end{figure}

We do so iteratively by first performing an $\epsilon_\mathrm{o}$-accurate bisection search between the point $v_i$ and the point of intersection $c$ of the ray $\vec{r}_i$ passing through $v_i$ and $v_0$ with the hyperplane \eqref{E:support_function_plane}. This gives us points $y_{ij}$ indexed by each step $j$ of the new bisection search. We then solve the following modified convex feasibility program \emph{for every} $y_{ij}$ for the given direction $r_{d_i}$:
\begin{subequations}\label{E:optimization_overapprox}
    \begin{align}
        \mini_{\mathbf{u},x_0} \quad &0 \\
        \st \quad & \mathbf{u} \in \U^{N_\delta}\\
        &G x_0 + H \mathbf{u} \in \K^{N_\delta+1} \\
        &\tr{r_{d_i}} x_0 = \tr{r_{d_i}} y_{ij}.
    \end{align}
\end{subequations}
%

Notice that we no longer fix $x_0$; rather, we implicitly look for a point on the set $\{ x_0 \mid \tr{r_{d_i}} x_0 = \tr{r_{d_i}} y_{ij}\} \cap \K$, when $y_{ij}$ varies, that is a feasible point. We also do not erode the constraints as we did before since our goal here is find an \emph{over-approximation} of the true kernel.


Once a feasible solution to \eqref{E:optimization_overapprox} is found for the desired accuracy $\epsilon_\mathrm{o}$ of the bisection search, we stop the iterations and store two entities:
\begin{enumerate}[(S1)]
    \item The last \emph{infeasible} step, i.e.\ the last value of $y_{ij}$ for which \eqref{E:optimization_overapprox} was infeasible. Denote this value by $y_{ij}^{\mathrm{inf}*}$; \label{list:overapprox_infeasible}
    \item The feasible solution pair $(\mathbf{u}_i^*,x_{0i}^*)$ (indexed by $i$ to correspond to the direction $r_{d_i}$). \label{list:overapprox_feasible}
\end{enumerate}

We first use entity~(S\ref{list:overapprox_infeasible}) to form our over-approximation along $r_{d_i}$ as the halfspace
\begin{equation}
    \{x \mid \tr{r_{d_i}} x \leq \tr{r_{d_i}} y_{ij}^{\mathrm{inf}*} \}.
\end{equation}
Clearly the set $\{x \mid \tr{r_{d_i}} x = \tr{r_{d_i}} y_{ij}^{\mathrm{inf}*} \}$ is a supporting hyperplane of $\Viab_\T(\K)$ with an arbitrary (and desirably conservative) error $\epsilon_\mathrm{o}$. By repeating the above procedure for all $N$ directions and forming the set
\begin{equation}\label{E:overapprox_VN}
    \widehat{\V}^{\zeta,\epsilon_\mathrm{o}}_N := \bigcap_{i=1}^N \{x \mid \tr{r_{d_i}} x \leq \tr{r_{d_i}} y_{ij}^{\mathrm{inf}*}\}
\end{equation}
we obtain (Fig.~\ref{fig:over_approx_all})
\begin{equation}
    \V^{\zeta, \epsilon}_N \subseteq \Viab_\T(\K) \subseteq \widehat{\V}^{\zeta, \epsilon_\mathrm{o}}_N.
\end{equation}
The error $\vol(\Viab_\T(\K)) - \vol(\V^{\zeta, \epsilon}_N)$ of our main under-approximation algorithm can then be quantitatively bounded above as $\vol(\widehat{\V}^{\zeta, \epsilon_\mathrm{o}}_N) - \vol(\V^{\zeta, \epsilon}_N)$. This upper-bound monotonically decreases as $N$ increases, and converges to a numerical constant as $\epsilon,\epsilon_\mathrm{o} \to 0$ and $N,\zeta \to \infty$.

\begin{figure}[t]
    \centering
    \scalebox{0.9}{\input{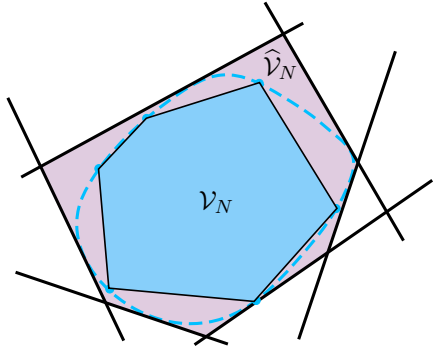}}
    \caption{The viability kernel is sandwiched in between the under-approximation $\V^{\zeta, \epsilon}_N$ and the over-approximation $\widehat{\V}^{\zeta, \epsilon_\mathrm{o}}_N$.}
    \label{fig:over_approx_all}
\end{figure}




The entity~(S\ref{list:overapprox_feasible}) from the feasibility program~\eqref{E:optimization_overapprox} can help us improve our under-approximation. This is discussed next.

\section{Improving the Under-Approximation}\label{S:improve_underapprox_bias}

We describe two techniques that help improve the quality of our under-approximation for finite number of samples. 

\subsection{Center of Mass \& Gauss-Kronecker Curvature}

If we were to naively sample from a uniform distribution on $\S_2^{n-1}$ to compute our under-approximating set, two elements would negatively affect the quality of such approximation: (a) the unbalanced distances between the unit ball centered at $v_0$ and the boundary points of the true kernel, i.e.\ the distance between $v_0$ and the center of mass of the kernel, and (b) the regions of the boundary of the kernel with high curvature. This is because a uniform distribution on $\S_2^{n-1}$ mapped onto $\partial \Viab_\T(\K)$ could yield a distribution that is far from uniform depending on severity of (a) or (b); Fig.~\ref{fig:gauss_distance_problem}.

\begin{figure}[t]
    \centering
    \scalebox{0.6}{\input{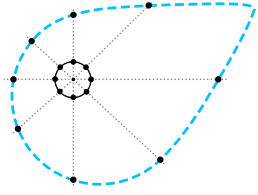}}
    \caption{Mapping a uniform distribution on $\S_2^{n-1} \subseteq \C$ linearly onto $\partial \C$ may result in a distribution that is non-uniform.}
    \label{fig:gauss_distance_problem}
\end{figure}

The former issue can be somewhat mitigated by continually moving the point $v_0$ to the center (e.g.\ in the sense of Chebyshev) of the set $\V_N^{\zeta,\epsilon}$ every time a new vertex is added. The latter issue can be addressed using the information obtained from the over-approximation procedure discussed in the previous section, specifically using the stored entity~(S\ref{list:overapprox_feasible}). The idea is that the greater the distance between the $i$th under-approximation vertex and over-approximation facet, the higher the Gauss curvature of the boundary of the true kernel in the neighborhood of that unexplored region.

To account for this problem, we perform an additional step after the $i$th iteration of our combined algorithm, every time a vertex is added to the under-approximation and an over-approximating halfspace is formed along the direction $r_{d_i}$: The stored value $x_{0i}^*$ in (S\ref{list:overapprox_feasible}) approximates the \emph{support vector} (the point in which the support function of a convex set touches its boundary) of the viability kernel in the direction $r_{d_i}$ with $\epsilon_\mathrm{o}$ accuracy. Therefore, we use $x_{0i}^*$ and execute a single instance of our under-approximation procedure this time along not $r_{d_i}$, but along the direction $x_{0i}^*- v_0$ (corresponding to the ray passing through $v_0$ and $x_{0i}^*$).

By doing so, we allow the over-approximation to ``guide'' where we should look for the next under-approximation vertex. Moreover, for this new instance of the under-approximation algorithm we can limit the bisection search to a diameter of roughly $2\epsilon_\mathrm{o}$ around $x_{0i}^*$ (instead of perfoming the search between $v_0$ and the point at which the new ray intersects $\partial\K$) since we know that $x_{0i}^*$ is already fairly close to the boundary of the true kernel. The resulting vertex is not only in close proximity to the point where the over-approximating halfplane at the $i$th iteration has been formed (thus providing a superior confidence that the viability kernel is sandwiched tightly in that area), but also covers the parts of the boundary that could potentially have a high Gauss curvature; Fig.~\ref{fig:over_approx_all_improved}.

\begin{figure}[t]
    \centering
    \scalebox{0.9}{\input{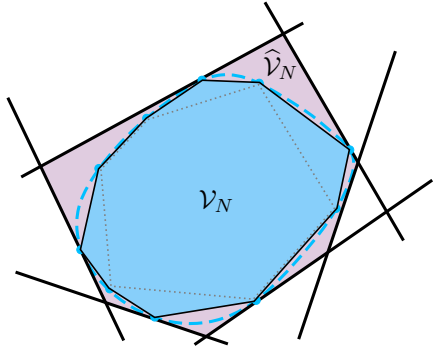}}
    \caption{Injecting additional under-approximation steps guided by the available over-approximation facets improves the quality of the resulting set (compare to Fig.~\ref{fig:over_approx_all}). Areas of the kernel with high curvature are now covered at a faster rate.}
    \label{fig:over_approx_all_improved}
\end{figure}


\subsection{Biased Random Sampling}\label{S:guiding}

The shortcomings of uniform sampling are more pronounced in high dimensions. Thus we additionally seek to \emph{bias} the distribution on the unit ball to mitigate these shortcomings. To this end, we present a few heuristic techniques that are still based on random sampling so as to keep the optimality results of Section~\ref{subsec:Algorithm_Convergence}, but could potentially improve the performance of the algorithm.

We will make use of the von-Mises Fisher (vMF) distribution \cite{Mardia_directionalstat} whose density function is given by
\begin{equation}
    f_{\mathrm{vMF}}(x; \mu, \kappa) := C(\kappa) e^{\kappa \mu^\top x}
\end{equation}
with \emph{concentration} $\kappa\geq 0$, mean direction $\mu$ ($\norm{\mu}=1$), and a normalizing constant $C(\kappa)$. The parameter $\kappa$  determines how samples drawn from this distribution are concentrated around the mean direction $\mu$. For $\kappa=0$ the vMF reduces to the uniform density; otherwise, it resembles a normal density (with compact support on the unit ball), centered at $\mu$ with variance inversely proportional to $\kappa$ (Fig.~\ref{fig:vMF_sample}). As $\kappa \to \infty$ the vMF converges to a point distribution.

\begin{figure}[t]
    \centering
    \scalebox{1}{\input{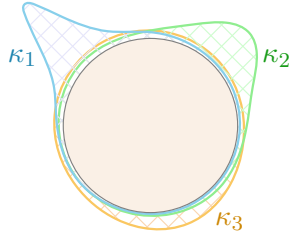}}
    \caption{The vMF density function on the unit ball for three different mean directions and concentrations with $\kappa_1 > \kappa_2 > \kappa_3$.}
    \label{fig:vMF_sample}
\end{figure}

The solutions we propose are by no means exhaustive, and there may be better ways of guiding the sampling process depending on the problem in hand or if we have some \emph{a priori} knowledge of the shape of the viability kernel.

\subsubsection{Gradient-Like Methods}

The first approach we discuss is related to how some measure of the error varies in consecutive steps of our combined under- and over-approximation algorithm. We let this change dictate from what vMF distribution should the next sample be drawn.

Let $\V^{\zeta,\epsilon}_i$ and $\widehat{\V}^{\zeta,\epsilon_\mathrm{o}}_i$ respectively denote the under-approximation and over-approximation of the viability kernel after adding the $i$th sample direction $r_{d_i}$, and let the state space be equipped with some metric $d$. Define the error function (possibly nonconvex) $\err \colon \Real^n \to \Real_{\geq 0}$ at this $i$th iteration as
\begin{equation}\label{E:error_dist}
    \err(r_{d_i}) := d(\V^{\zeta,\epsilon}_i, \, \widehat{\V}^{\zeta,\epsilon_\mathrm{o}}_i).
\end{equation}
Ideally, the metric $d$ is chosen such that $\err(r_{d_i})$ is monotonically non-increasing as $i$ increases (i.e.\ as we produce more accurate approximations). Define the normalized quantity
\begin{equation}
    \nabla\err(r_{d_i}) := \frac{1 -  d(\V^{\zeta,\epsilon}_i, \, \widehat{\V}^{\zeta,\epsilon_\mathrm{o}}_i)/ d(\V^{\zeta,\epsilon}_{i-1}, \, \widehat{\V}^{\zeta,\epsilon_\mathrm{o}}_{i-1})}{\frac{1}{\pi} \cos^{-1} \left( \inner{r_{d_i}, r_{d_{i-1}}} \right)},
\end{equation}
treating the pathological case $\frac{0}{0}$ as $0$. The magnitude of $\nabla\err(r_{d_i})$ approximates the rate of improvement between two consecutive iterations of our combined algorithm. Thus, we can use this information to draw the next sampling direction $r_{d_{i+1}}$ as
\begin{equation}
    r_{d_{i+1}} \sim f_{\mathrm{vMF}}(x; \mu_{i+1}, \kappa_{i+1})
\end{equation}
with mean and concentration parameters updated as
\begin{align}
        \kappa_{i+1} &= \max\left\{ \nu_0,\, \omega_i \right\},\\
        \mu_{i+1} &= \sgn( -\nu_0 + \omega_i) \cdot r_{d_i},
\end{align}
where
\begin{equation}
    \omega_i := \nu_1 \tanh\left( \nu_2 \nabla\err(r_{d_i}) \right).
\end{equation}
Here the positive scalars $\nu_0$, $\nu_1$, and $\nu_2$ are design choices; for example, $\nu_2$ dictates how quickly the tangent hyperbolic function reaches its maximum $\nu_1$ where it levels off, while $\nu_0$ (which is a small scalar) limits how close the resulting vMF distribution can be to the uniform one and, when necessary, sets the sign of the direction $\mu_{i+1}$ of the distribution to the negative of $r_{d_i}$.

The intuition behind such an update rule is that when sampling along a given direction $r_{d_i}$ causes the magnitude of $\nabla \err$ to become large (thus causing the error to decrease quickly), drawing the next sample from a distribution that is concentrated (proportional to this change) around the same direction may be a good choice to continue to reduce the error: $\kappa_{i+1} = \omega_i$ and $\mu_{i+1} = r_{d_i}$. If, on the other hand, the magnitude of $\nabla \err$ is small, we are approaching a local minima of the quasilinear function $\err(\cdot)$. Therefore, we would want to sample from a distribution that is \emph{not} concentrated around the current direction $r_{d_i}$, and could also be such that it potentially minimizes the likelihood of sampling around $r_{d_i}$. In such a case, we would bias the sampling to be slightly concentrated around the opposite direction (and also decrease the likelihood of resampling around $r_{d_i}$) so as to attempt to get us out of this local minima, while not making it to be too far off from a uniform distribution: $\kappa_{i+1} = \nu_0$ and $\mu_{i+1} = -r_{d_i}$.

%
%
%
%
%
%



\textbf{Volumetric Measure:} We can use volume as a metric in \eqref{E:error_dist} so that $\err(r_{d_i}) := \vol(\widehat{\V}^{\zeta,\epsilon_\mathrm{o}}_i) - \vol(\V^{\zeta,\epsilon}_i)$. Of course, computing the volume of a polytope is $\#$P-hard \cite{Dyer_vol_NPhard88}. However, we can conservatively approximate $\err(r_{d_i})$ by calculating analytically the volume of appropriate inscribed and circumscribed ellipsoids: A minimum volume circumscribed ellipsoid (mVCE) containing a polytope that is represented by its vertices can be computed efficiently via a semidefinite program \cite{Boyd_cvxbook}. The same is true for a maximum volume inscribe ellipsoid (MVIE) that is contained in a polytope represented by its facets. (Note that the converse problems are NP-hard.) If we computed the mVCE of $\V^{\zeta,\epsilon}_i$ and shrunk it by a factor of $n$ (the dimension), then the resulting ellipsoid would be a subset of $\V^{\zeta,\epsilon}_i$. Similarly, if we enlarged the MVIE of $\widehat{\V}^{\zeta,\epsilon_\mathrm{o}}_i$ by a factor of $n$, then $\widehat{\V}^{\zeta,\epsilon_\mathrm{o}}_i$ would be a subset of the resulting ellipsoid. The difference in the volume of these two ellipsoids would be an upper-bound on $\err(r_{d_i})$.

Working with the volume of these extremal ellipsoids does preserve order, in that the relations $\vol(\V^{\zeta,\epsilon}_{i+1}) \geq \vol(\V^{\zeta,\epsilon}_i)$ and $\vol(\widehat{\V}^{\zeta,\epsilon_\mathrm{o}}_{i+1}) \leq \vol(\widehat{\V}^{\zeta,\epsilon_\mathrm{o}}_i)$ are also true for their extremal ellipsoids. The reason is that adding constraints (additional vertices or facets) to a convex optimization problem (the SDPs) cannot decrease the value of the optimal objective functions (correlated with ellipsoid volumes). As a result, the error upper-bound, just like the error itself, is monotonically non-increasing as $i$ increases. A disadvantage of working with scaled extremal ellipsoids is that the shrinkage/enlargement by a factor of $n$ can be too conservative particularly for larger $n$.

\textbf{Hausdorff Distance:} The Hausdorff distance between two compact convex sets $\C_1$ and $\C_2$ in $\Real^n$ in terms of their support functions is defined as
\begin{equation}
    d_\mathrm{H}(\C_1,\C_2) := \max_{\ell \in \S^{n-1}} \left\{ \abs{\rho_{\C_1}(\ell) - \rho_{\C_2}(\ell)} \right\}.
\end{equation}
We can use this metric to define the error: $\err(r_{d_i}) = d_\mathrm{H}(\V^{\zeta,\epsilon}_i, \widehat{\V}^{\zeta,\epsilon_\mathrm{o}}_i)$. Computing this distance between $\V^{\zeta,\epsilon}_i = \conv(\{v_j\}_{j=0}^i)$ and $\widehat{\V}^{\zeta,\epsilon_\mathrm{o}}_i =: P_i x\leq p_i$ (where $P$ and $p$ respectively are the appropriate matrix and vector corresponding to the facet-based outer polytope defined in \eqref{E:overapprox_VN} but with $i$ faces) can be cast as a series of LPs: For a fixed direction $\ell \in \S^{n-1}$, we compute at every iteration $i$ 
\begin{subequations}
\begin{align}
    \rho_{\V^{\zeta,\epsilon}_i}(\ell) = \max_{x,\, \{\lambda_j \geq 0\}} \quad &\ell^\top x\\
    \st \quad &x = \sum_{j=0}^i \lambda_j v_j,  \:\;  \sum_{j=0}^i \lambda_j = 1,
\end{align}
\end{subequations}
as well as
\begin{subequations}
\begin{align}
    \rho_{\widehat{\V}^{\zeta,\epsilon_\mathrm{o}}_i}(\ell) = \max_{x} \quad &\ell^\top x\\
    \st \quad &P_i x \leq p_i.
\end{align}
\end{subequations}
The Hausdorff distance can then be approximated over a finite number of directions $\Lom \subset \S^{n-1}$ as
\begin{multline}\label{E:Hausdorff_approx}
    \hat{d}_\mathrm{H}(\V^{\zeta,\epsilon}_i, \widehat{\V}^{\zeta,\epsilon_\mathrm{o}}_i) := \max_{\ell \in \Lom} \left\{ \abs{\rho_{\widehat{\V}^{\zeta,\epsilon_\mathrm{o}}_i}(\ell) - \rho_{\V^{\zeta,\epsilon}_i}(\ell)} \right\}\\
    \leq \max_{\ell \in \S^{n-1}} \left\{ \abs{\rho_{\widehat{\V}^{\zeta,\epsilon_\mathrm{o}}_i}(\ell) - \rho_{\V^{\zeta,\epsilon}_i}(\ell)} \right\}
    = d_\mathrm{H}(\V^{\zeta,\epsilon}_i, \widehat{\V}^{\zeta,\epsilon_\mathrm{o}}_i).
\end{multline}
This approximation is not conservative due to the inequality in \eqref{E:Hausdorff_approx}. But as far as guiding the sampling, it may yield a viable alternative to the volumetric error described above (that could be excessively conservative). Unfortunately, performing these additional LPs at the end of each iteration of our algorithm may undermine its efficiency.



\subsubsection{Purely Heuristic Methods}\label{S:AOD_PtoP}

We also propose two purely heuristic methods as alternatives to the gradient-like methods presented above.

\textbf{Averaged Opposite Direction:} For this approach we associate a given vMF density function to the negative of \emph{each} individual direction vector we have generated so far, and draw our next sampling direction from a convolution of these density functions. That is, in some sense we are drawing at random a sample whose expected value lies in the opposite direction of the samples we have already drawn. The intuition here is that we can approach ``true'' uniformity at a higher rate by sampling in the direction whose neighborhood we have not yet sampled as densely.

We can vary the concentration of the vMF densities as an increasing function of the iteration step $i$ (e.g.\ linearly with the number of vertices), such that at the beginning steps of our algorithm these densities are closer to uniform, and as we progress and generate more vertices/directions their concentrations increase.\footnote{The density function of the convolution is closer to uniform if (a) the individual densities are closer to uniform, or (b) the existing directions all cancel each other out.} A design parameter $\nu_1$ can be a multiplier for the number of vertices. We cap the concentration of the distributions to 100 so as to not lose the randomness properties of the algorithm when $i$ grows too large. 

\textbf{Point-to-Plane Distance:} This approach is based on simply calculating the smallest distance from every vertex of $\V^{\zeta,\epsilon}_i$ to facets of $\widehat{\V}^{\zeta,\epsilon_\mathrm{o}}_i$, and identifying the largest of such distances. This quantity in some sense describes the gaps between the two sets over which we are most uncertain about the boundary of the true kernel.

Let $\bar{v}_j$ be the vertex in $\V^{\zeta,\epsilon}_i = \conv(\{v_j\}_{j=0}^{i})$ with the largest point-to-plane distance. To guide the sampling at the next iteration of the algorithm, we calculate the ray that passes through the center $v_0$ and the point in the facet of $\widehat{\V}^{\zeta,\epsilon_\mathrm{o}}_i$ that is the closest to $\bar{v}_j$. To determine our next sampling direction $r_{d_{i+1}}$, we then sample from a vMF distribution whose mean is the unit vector along this ray and whose concentration is dependant on the calculated point-to-plane distance from $\bar{v}_j$ to $\widehat{\V}^{\zeta,\epsilon_\mathrm{o}}_i$. A design parameter can again be a multiplier $\nu_1$ for this distance.

\subsubsection{Calibration and Comparison on Random Systems}\label{S:calibration_testing}

To compare the performance of the four guiding methods described above, we first roughly calibrated the design parameters for each approach by running each guiding technique on 30 randomly generated systems across 2D, 3D, and 4D state and input dimensions\footnote{Our tests were limited by 5D due to the need to directly compute the volume of polytopes for performance assessment.} for a variety of exponents of 10. For example, we tested the volume extremal ellipsoids and the Hausdorff distance methods for all combinations of $\nu_0=10^{\{-1,0\}}$, $\nu_1 = 10^{\{0, 1, 2\}}$, and $\nu_2 = 10^{\{0,1\}}/n$ (made dimension $n$ dependent so that the tangent hyperbolic reaches its maximum $\nu_1$ more slowly when $n$ is larger); and the averaged opposite direction and the point-to-plane distance methods for $\nu_1 = 10^{\{-2,-1,0,1\}}/n$ (made dimension dependent so that the vMF concentration increases more slowly when $n$ is larger). We found the optimal parameters to be $\nu_0=0.1$, $\nu_1 = 1$, $\nu_2 = 1/n$ for the volume extremal ellipsoids; $\nu_0=1$, $\nu_1 = 1$, $\nu_2 = 1/n$ for the Hausdorff distance; $\nu_1=0.1/n$ for the averaged opposite direction; and $\nu_1=1/n$ for the point-to-plane distance. We emphasize that this calibration was meant only to achieve parameters that were roughly the correct order of magnitude; we did no further fine tuning. 

With the calibrated methods in hand, we then ran all four methods on a fresh set of 60 randomly generated systems across 3D, 4D, and 5D state and input dimensions. We were then able to evaluate the performance of
each method by taking advantage of the fact that all the approximations
generated are under-approximations.  Thus, the under-approximation with
the largest possible volume must be (in a volumetric sense) the closest
to the true viability kernel.  Using this fact we found for each random
system and for a given number of vertices the best under-approximation
selected from the results of all four different methods (as well as the
basic uniform sampling method).  Then, for each system and number of
vertices we calculated the percent difference in volume between the best
under-approximation and the result of each method.  (Note that this
percent difference must be negative, since the result from each method
must have a volume smaller than or equal to the result of the best
method.)  Thus a method with a less negative value, i.e. closer to 0,
indicates a method that performs better, in a volumetric sense, than a
method with a more negative value (which indicates a much smaller and
thus less accurate under-approximation in a volumetric sense).  Fig.~\ref{F:heuristic_comparison}
plots the resulting percent volume difference by method as a function of
the number of vertices, averaged over the 60 randomly generated systems.


On average, the averaged opposite direction (blue) outperforms all methods. More specifically, it improves the resulting under-approximation volume by about 5\% over the uniform sampling (red). The Hausdorff distance (yellow) performs consistently better than uniform sampling but since we only compute a crude approximation of the actual Hausdorff distance in \eqref{E:Hausdorff_approx} (at every step we only use the directions along which we have already generated a vertex so as to maintain monotonicity of the distance approximate while keeping the computation times on par with the other methods), this improvement is expected to be more emphasized for higher number of samples than only 50. The volume extremal ellipsoids (green) initially performs far better than uniform sampling in 3D and 4D, but then quickly degrades in 5D due to shrinkage/expansion of the ellipsoids by a factor of $n$ that is much more exaggerated in higher dimensions. The point-to-plane distance (cyan) simply does not show any improvement over uniform sampling, at least on average in our test setup. We will use the averaged opposite direction as our guiding method of choice for our examples in Section~\ref{S:examples}. However, we do emphasize again that the performance of these heuristics is problem-dependent. Therefore, while in our limited random tests one method might have outperformed others on average, it may be that another method is even more suitable and performs significantly better for a given problem.

\begin{figure}[t]
    \centering
    \includegraphics[clip = true, scale=0.55]{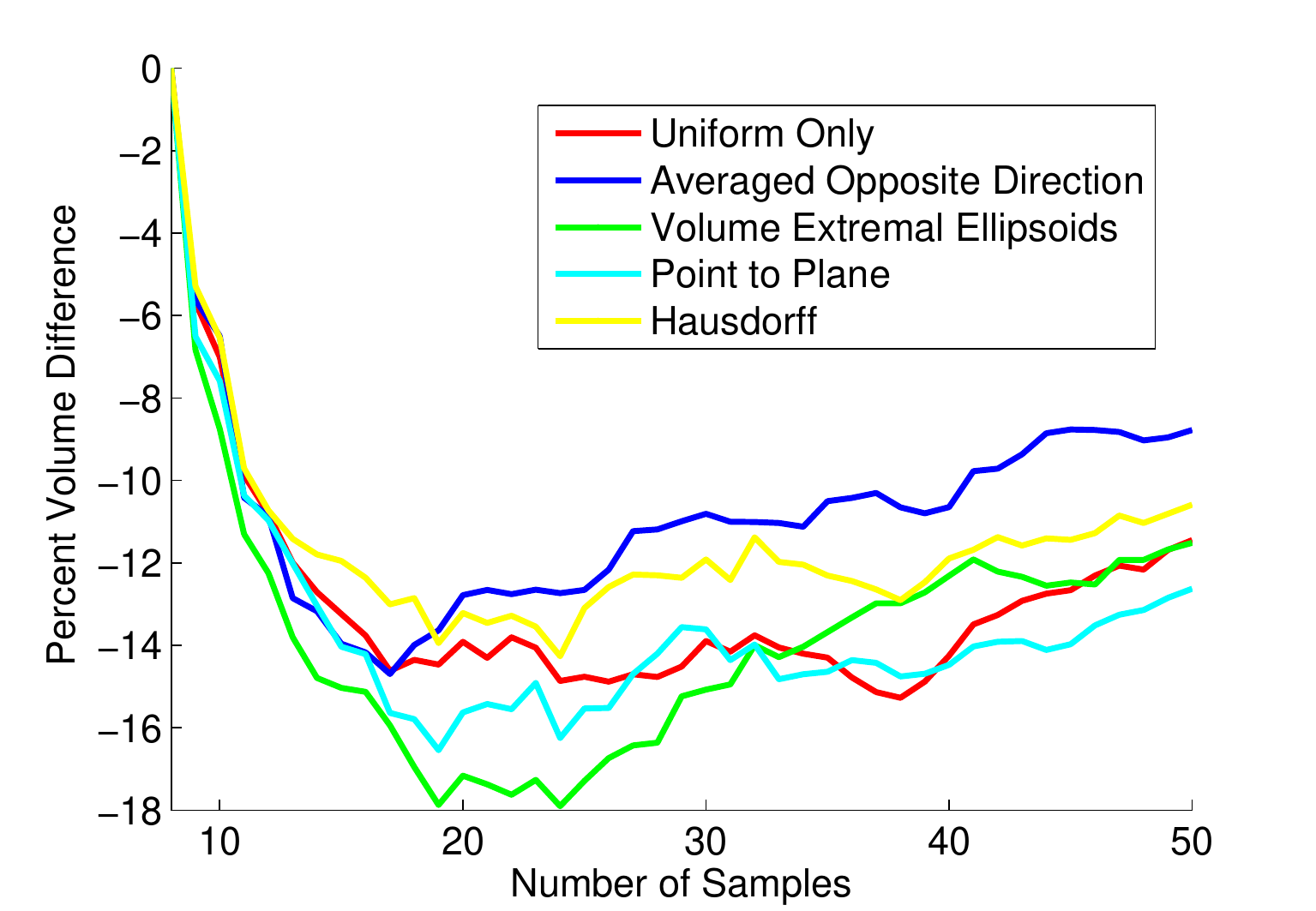}
    \caption{Percent difference in volume between the result of each guiding method and the best result out of all the guiding methods for a given number of vertices (averaged over 60 randomly generated systems across 3, 4, and 5 state and input dimensions).}
    \label{F:heuristic_comparison}
\end{figure}

\section{Examples}\label{S:examples}

\subsection{The Double-integrator}

First, consider the simple dynamics $\ddot{x}=u$ with $\delta = 0.05\,\text{s}$. The constraints $\K = \{ x\mid \norm{x}_\infty \leq 0.5\}$ and $\U = [-0.15,0.15]$ are to be respected over $\T=[0,1]$. To find a bound $M$ on the vector field and construct the eroded set $\K_\downarrow(M,\delta)$ we use the following procedure: We first scale the state space by computing for every dimension $d$ the quantity $z^*(\mathbf{e}_d) - z^*(-\mathbf{e}_d)$ with $z^*(\mathbf{e}_d) := \argmax_z \mathbf{e}_d^\top z$ subject to $z=\dot{x}$ and $x\in\K$, $u\in \U$, and where $\mathbf{e}_d$ denotes the standard basis vector spanning $d$th dimension. Then we divide all such quantities by their minimum value among all dimensions, and transform the dynamics accordingly so that a state increment in all dimensions is equivalent. At this stage, we can calculate $M$ as the optimal value of $\max_{x,u}\norm{\dot{x}}_\infty$ subject to $x\in\K$ and $u\in\U$. The eroded set $\K_\downarrow(M,\delta)$ is now constructed in the scaled state space according to Lemma~\ref{Lem:HSCC}.


We used $\zeta=4$th order discretization and $\epsilon = \epsilon_\mathrm{o} = 0.01$-accurate bisection search to obtain the under- and over-approximations shown in Fig.~\ref{F:integrator}. The under-approximation is computed using $N=20$ randomly generated samples via Algorithm~\ref{Alg:polytopic_approx}. The sampling process was guided via the techniques in Section~\ref{S:improve_underapprox_bias}. The over-approximation polytope consists of $N/2$ facets and is computed according to Section~\ref{S:over_approx}. The overall computation time was $18\,\text{s}$.

\begin{figure}[t]
    \centering
    \includegraphics[clip = true, scale = 0.55]{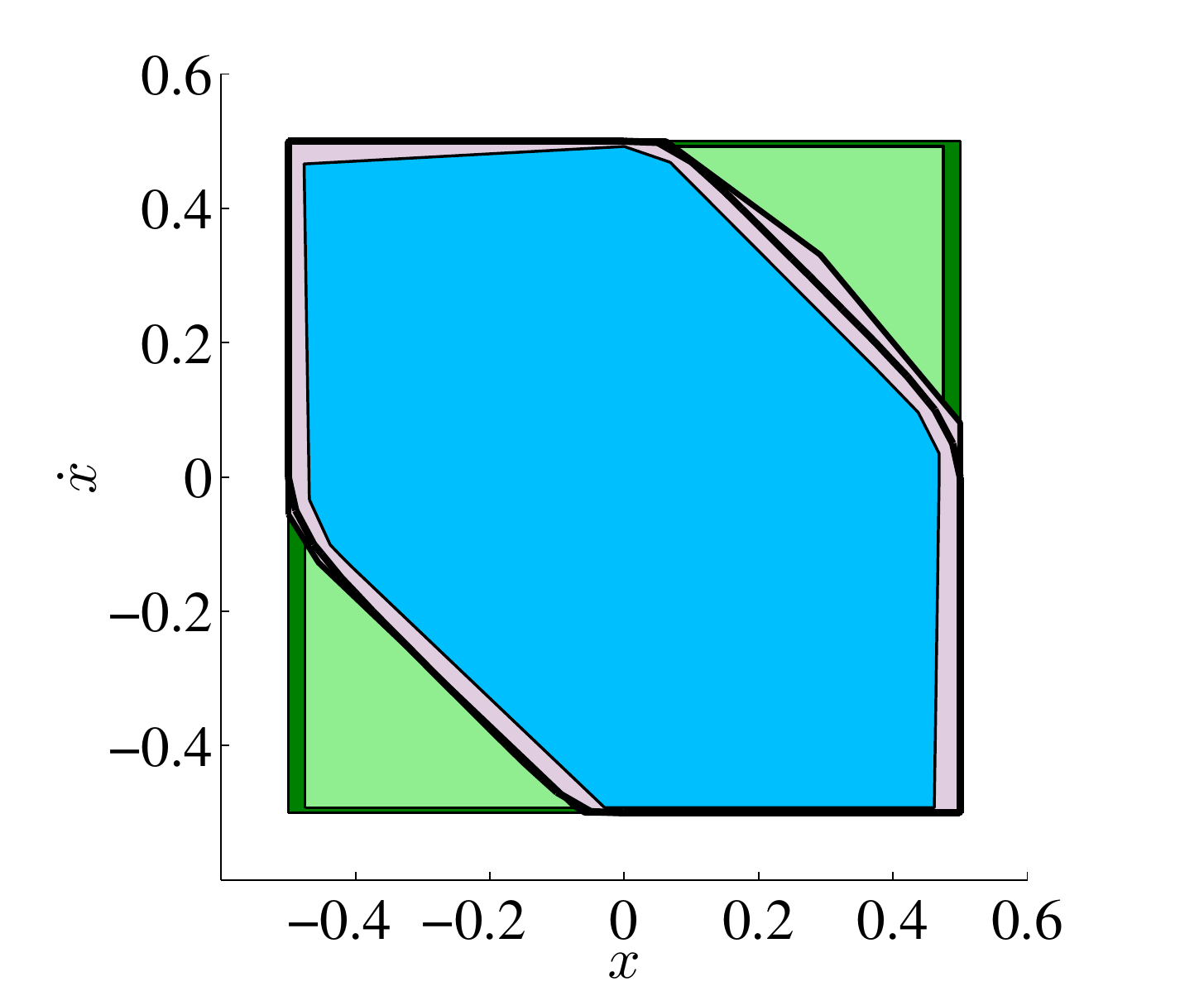}
    \caption{Polytopic under-approximation $\V^{\zeta,\epsilon}_N$ (blue) and over-approximation $\widehat{\V}^{\zeta,\epsilon_\mathrm{o}}_{N/2}$ (lavender) of $\Viab_\T(\K)$ for the double-integrator example with $N=20$ samples. The sets $\K$ and $\K_\downarrow(M,\delta)$ are shown in dark and light green, respectively. The SD level-set approximation ~\cite{kaynama_NAHS2012} is also shown (outlined in thick black line).} \label{F:integrator}
\end{figure}


\subsection{12D Quadrotor Flight Envelope Protection}

We now evaluate our algorithm on the benchmark example described in \cite{kaynama_benchmark}. Consider the full-order model of a quadrotor based on the nonlinear Newton-Euler rigid body equations of motion. The state vector
\begin{equation}
    x = \tr{\begin{bmatrix}
    \mathrm{x} & \mathrm{y} & \mathrm{z} & \dot{\mathrm{x}} & \dot{\mathrm y} & \dot{\mathrm z} & \phi & \theta & \psi & \dot{\phi} & \dot{\theta} & \dot{\psi}
    \end{bmatrix}} \in \Real^{12}
\end{equation}
is comprised of translational positions in $[\text{m}]$ with respect to a global origin, their derivatives (linear velocities in $\mathrm x$, $\mathrm y$, $\mathrm z$ directions) in $[\text{m}/\text{s}]$, the Eulerian angles roll $\phi$, pitch $\theta$, and yaw $\psi$ in $[\text{rad}]$, and their respective derivatives (angular velocities) in $[\text{rad}/\text{s}]$. The control input is the vector $u = \left[u_1 \; u_2 \; u_3 \; u_4 \right]^\top \in \Real^4$ consisting, respectively, of the total thrust in $[\text{m}/{\text{s}^2}]$ normalized with respect to the mass of the quadrotor ($u_1$) and the second-order derivatives $\ddot{\phi}$, $\ddot{\theta}$, $\ddot{\psi}$ of the Eulerian angles in $[\text{rad}/{\text{s}^2}]$ ($u_2$ through $u_4$). The system is under-actuated since there are six degrees of freedom but only four actuators. By linearizing the equations of motion about the hover condition $\phi = 0$, $\theta = 0$, and $u_1 = g$ (with $g \approx 9.81$ being the acceleration of gravity) one would obtain the model $\dot{x} = Ax+Bu$ with the state and the input now representing \emph{deviation} from the equilibrium. We follow the example detailed by~\cite{Cowling_quadrotor_paper} of an agile quadrotor in which the state is sampled at a frequency of $10 \,\text{Hz}$. (See the same reference for values of the system matrices $A$ and $B$.)


For safe operation of the vehicle the Eulerian angles $\phi$ and $\theta$ and the speed profile $V := \norm{\left[ \dot{\mathrm{x}} \; \dot{\mathrm y} \; \dot{\mathrm z} \right]}$ are bounded as $\phi,\theta \in [-\frac{\pi}{4}, \frac{\pi}{4}]$ and $V \leq 5$. The angular velocities are constrained as $\dot{\phi},\dot{\theta},\dot{\psi} \in [-3,3]$. We further assume that the vehicle must safely fly within the range of $1$ to $7 \, \text{m}$ above the ground in $\textrm{z}$ direction in an environment that stretches $6\,\text{m}$ in each direction in the $\mathrm{x}$-$\mathrm{y}$ plane. These constraints form the flight envelope $\K$. The vector $u$ is constrained by the hyper-rectangle $\U:=[-g, 2.38] \times [-0.5, 0.5]^3$. 

The quadrotor can travel a distance of roughly half a meter between any two consecutive sampling times despite the relatively high sampling frequency. This fact further warrants the treatment of such a safety-critical system through a SD framework. We wish to compute the set of initial states for which safety can be maintained over $\T=[0,2]$.


We warm-start our approximation algorithm (in the scaled state space) by first sampling along all axes in order to obtain a full-dimensional object, and then along 72 uniformly-spaced vectors in $x_i$-$x_{i+3}$ and $x_{i+6}$-$x_{i+9}$ subspaces for $i=1,2,3$. The remaining samples are generated randomly by guiding the vMF sampler via the averaged opposite direction method. The discretization order and the bisection search accuracies are the same as in the previous example.


Fig.~\ref{F:quad12D} shows selected 2D projections of our sampling-based polytopic approximations of $\Viab_\T(\K)$ in the original unscaled state space. The algorithm requires about $3\,\text{s}$ (without optimizing the code for speed) to generate a new vertex of the under-approximation and $5\,\text{s}$ to generate a facet of the over-approximation. The under-approximation is tight in the sense that each vertex of the polytope belongs to the boundary of the true viability kernel with some \emph{a priori} known accuracy due to the SD nature of the problem. The over-approximation is also tight in that each facet of the polytope touches the boundary of the true kernel in at least one point. The two approximating sets sandwich the boundary of the viability kernel to within a certain precision in at least half of the number of vertices of the under-approximation, providing an added layer of confidence about the precise location of the kernel. The tightness of the sets are unfortunately unobservable in the projection plots.


\newcommand{\constt}{0.25}
\newcommand{\spcc}{\hspace{0pt}}
\begin{figure*}[t]
    \centering
    \includegraphics[clip = true, width=\constt\textwidth]{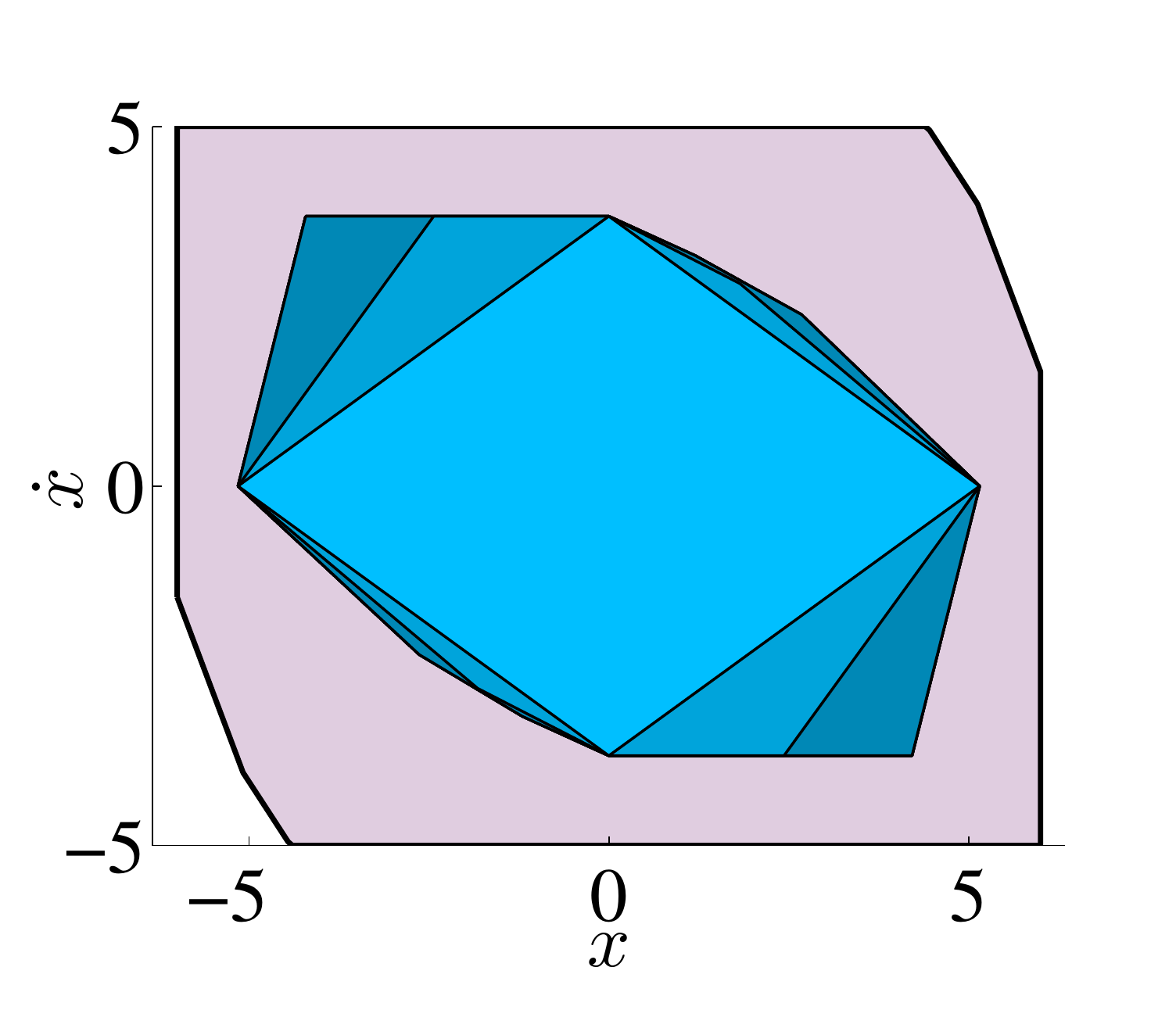}\spcc
    \includegraphics[clip = true, width=\constt\textwidth]{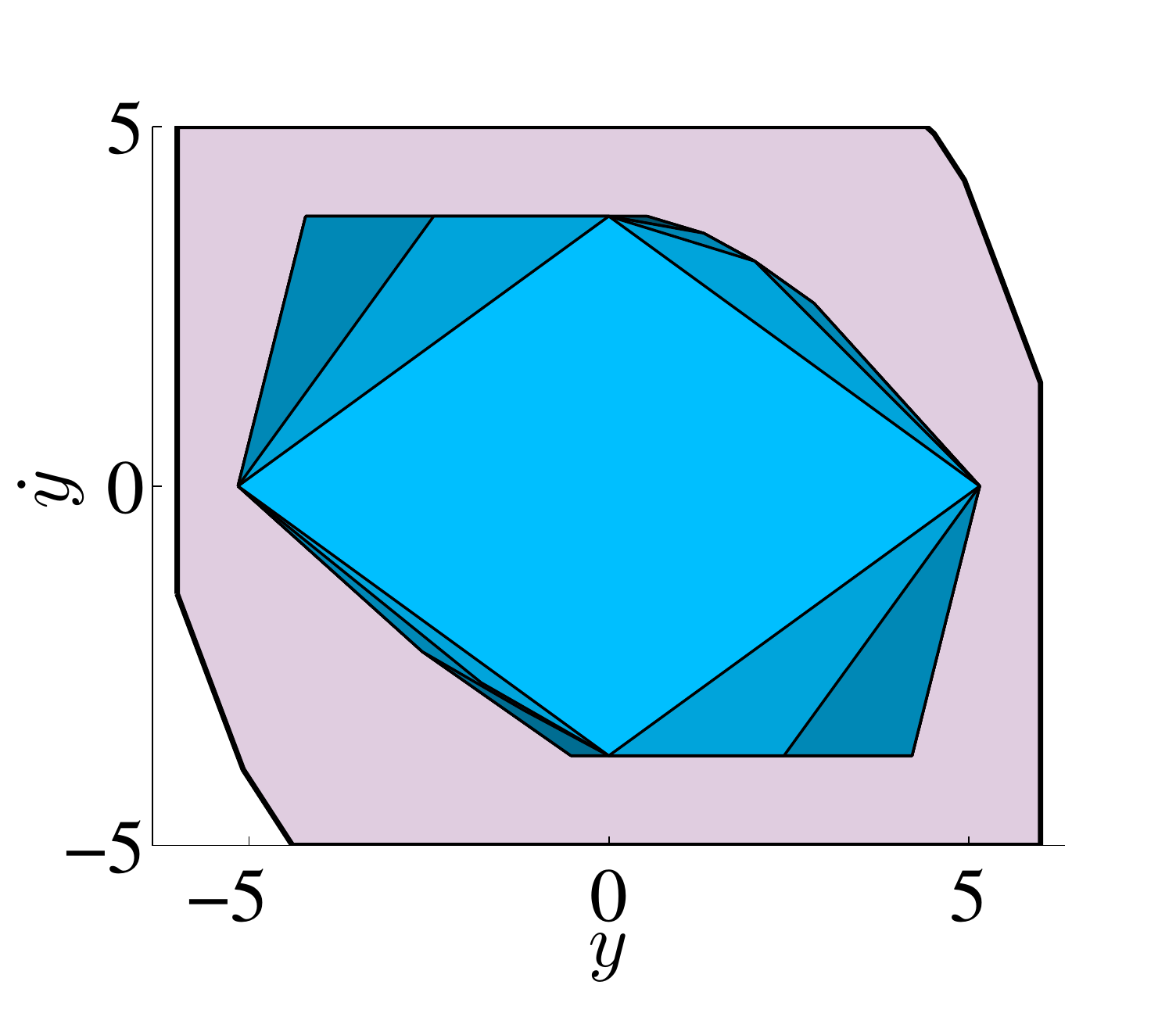}\spcc
    \includegraphics[clip = true, width=\constt\textwidth]{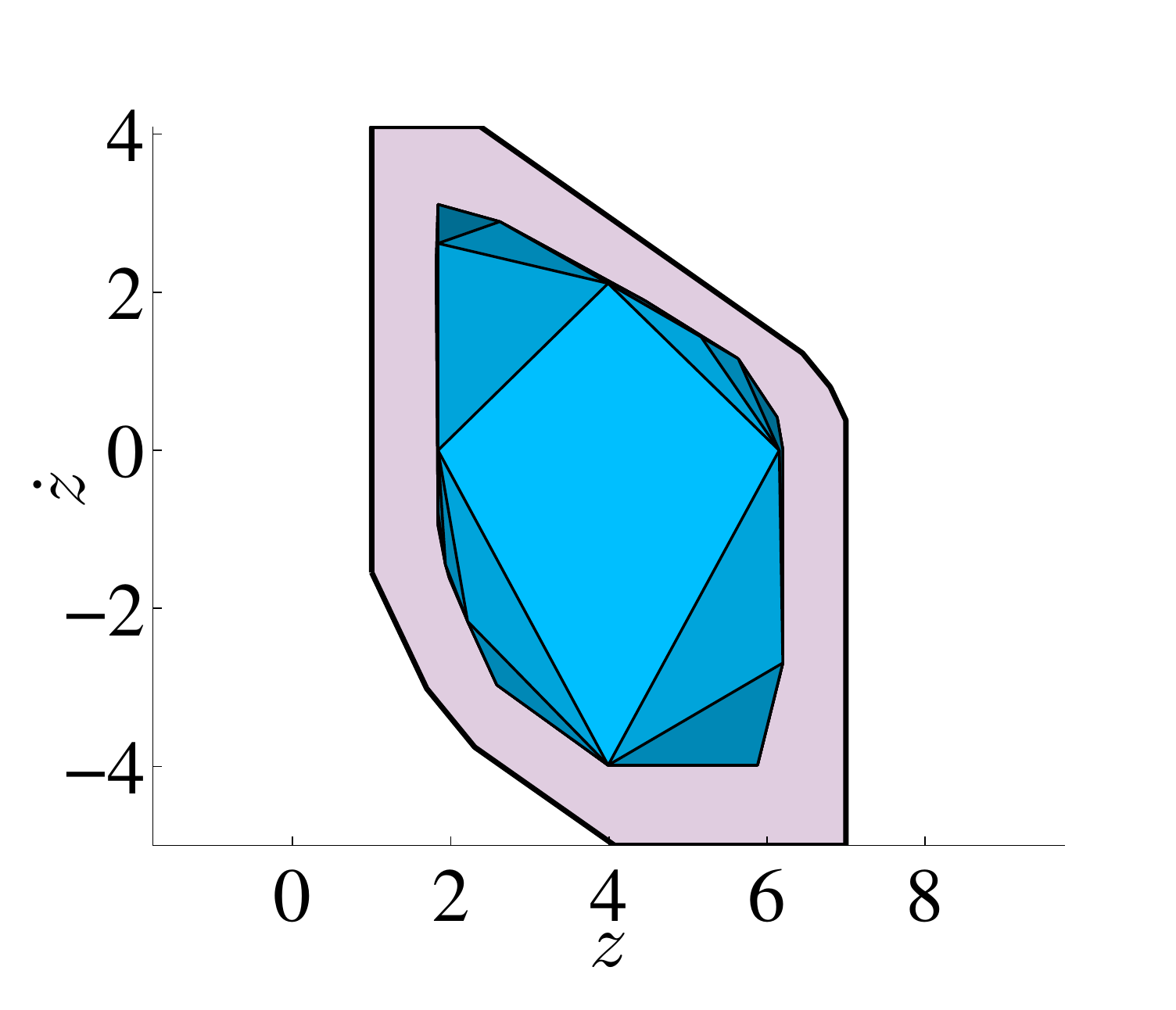}\spcc
    \includegraphics[clip = true, width=\constt\textwidth]{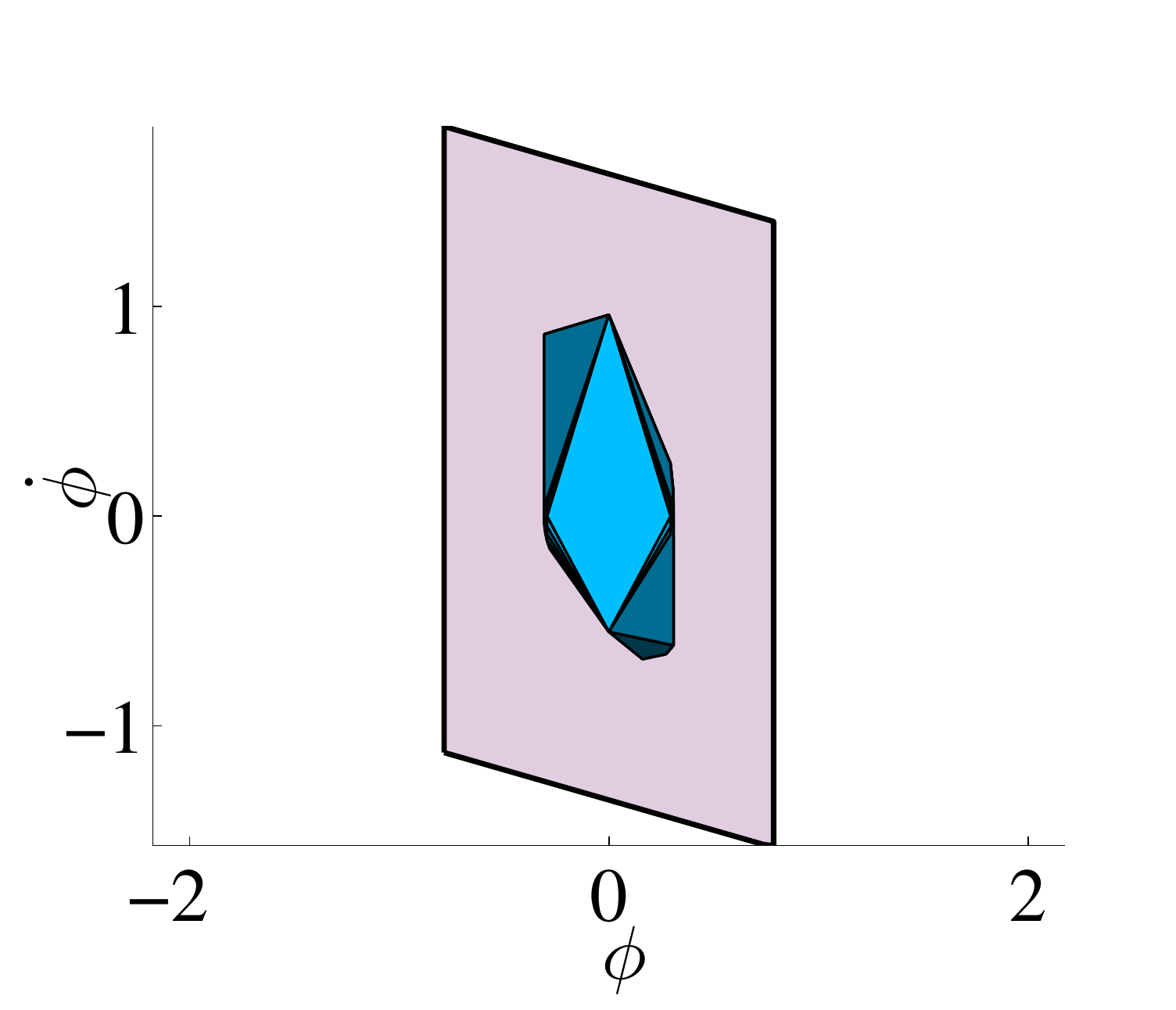}\\
    \includegraphics[clip = true, width=\constt\textwidth]{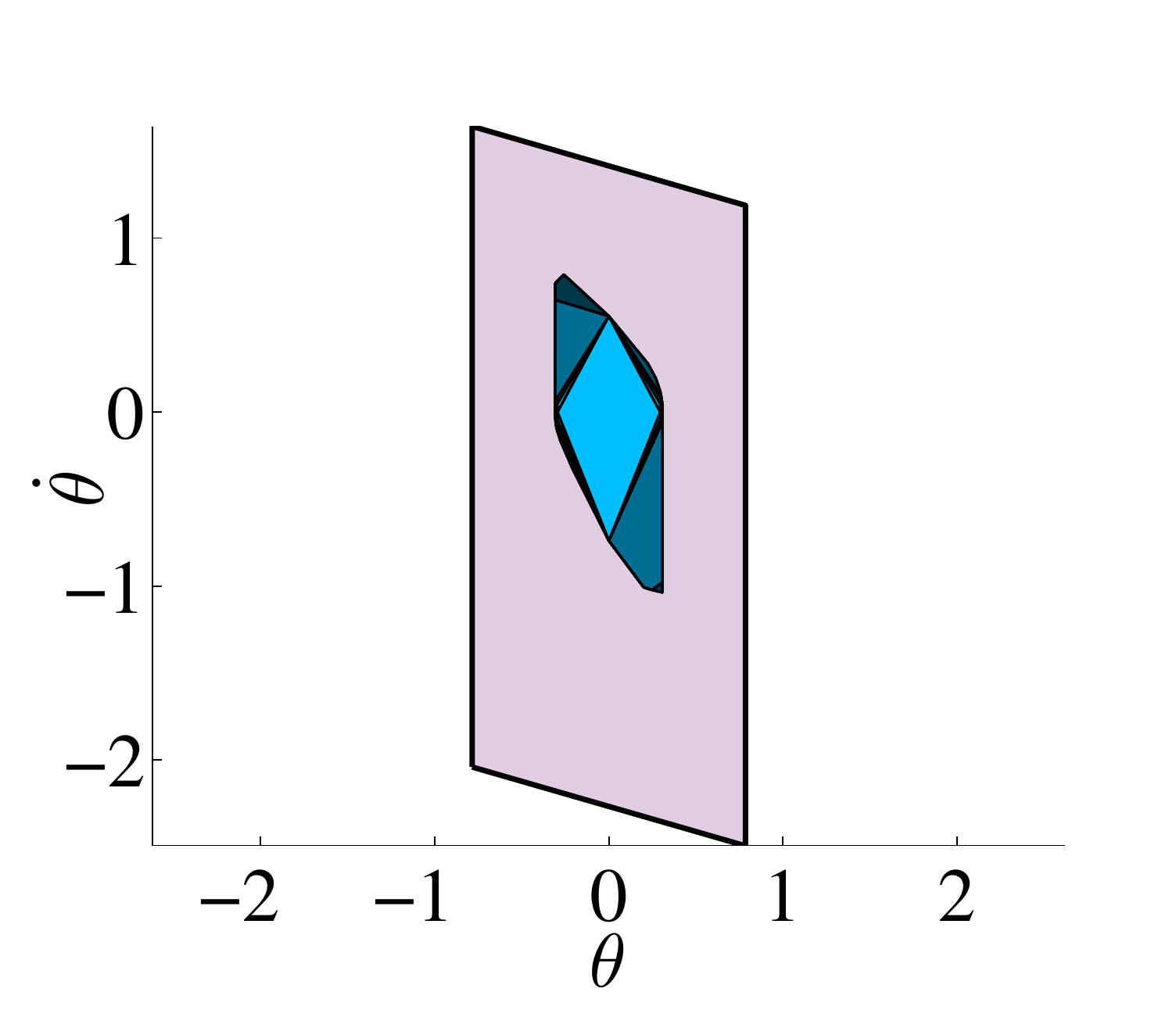}\spcc
    \includegraphics[clip = true, width=\constt\textwidth]{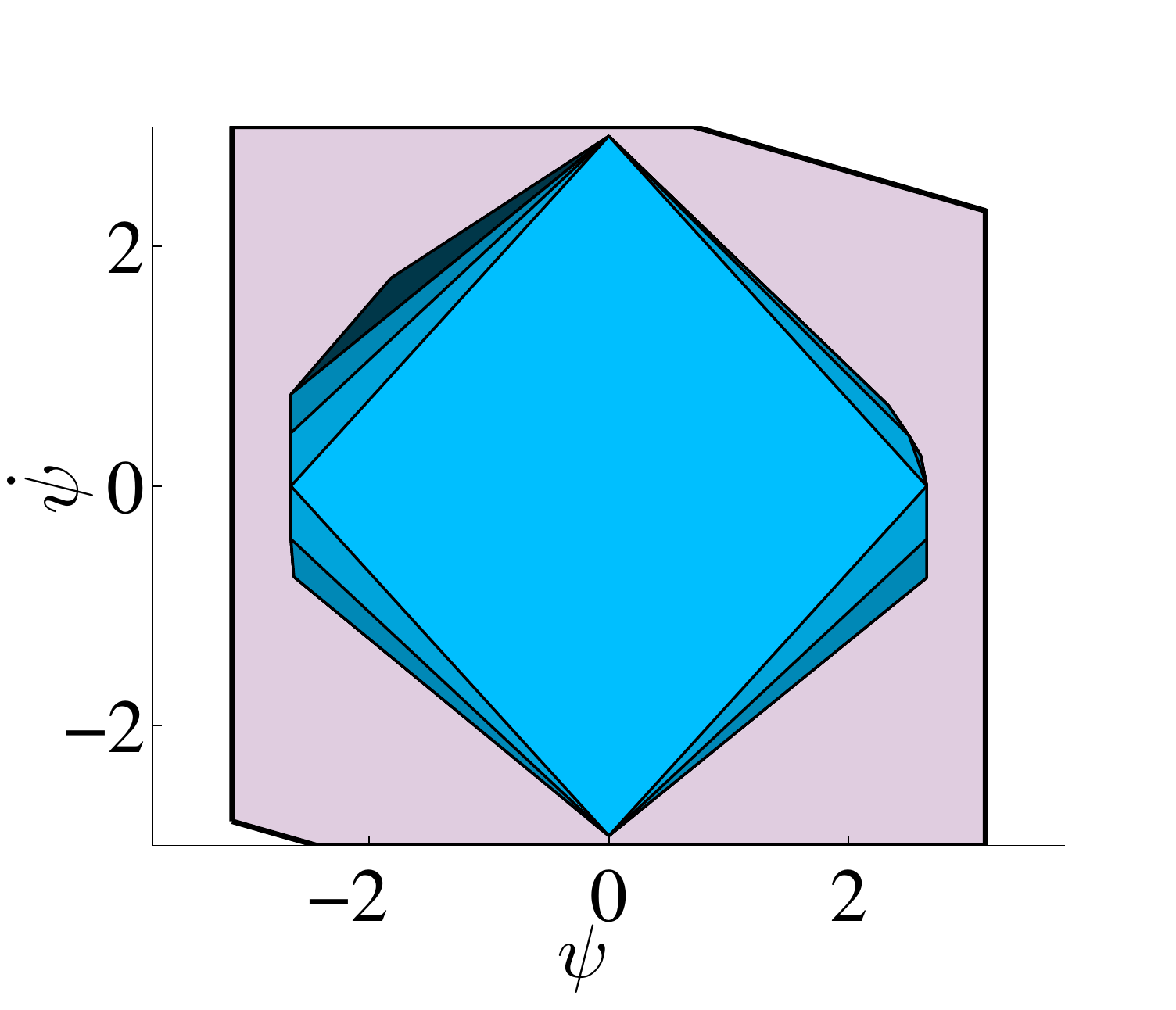}\spcc
    \includegraphics[clip = true, width=\constt\textwidth]{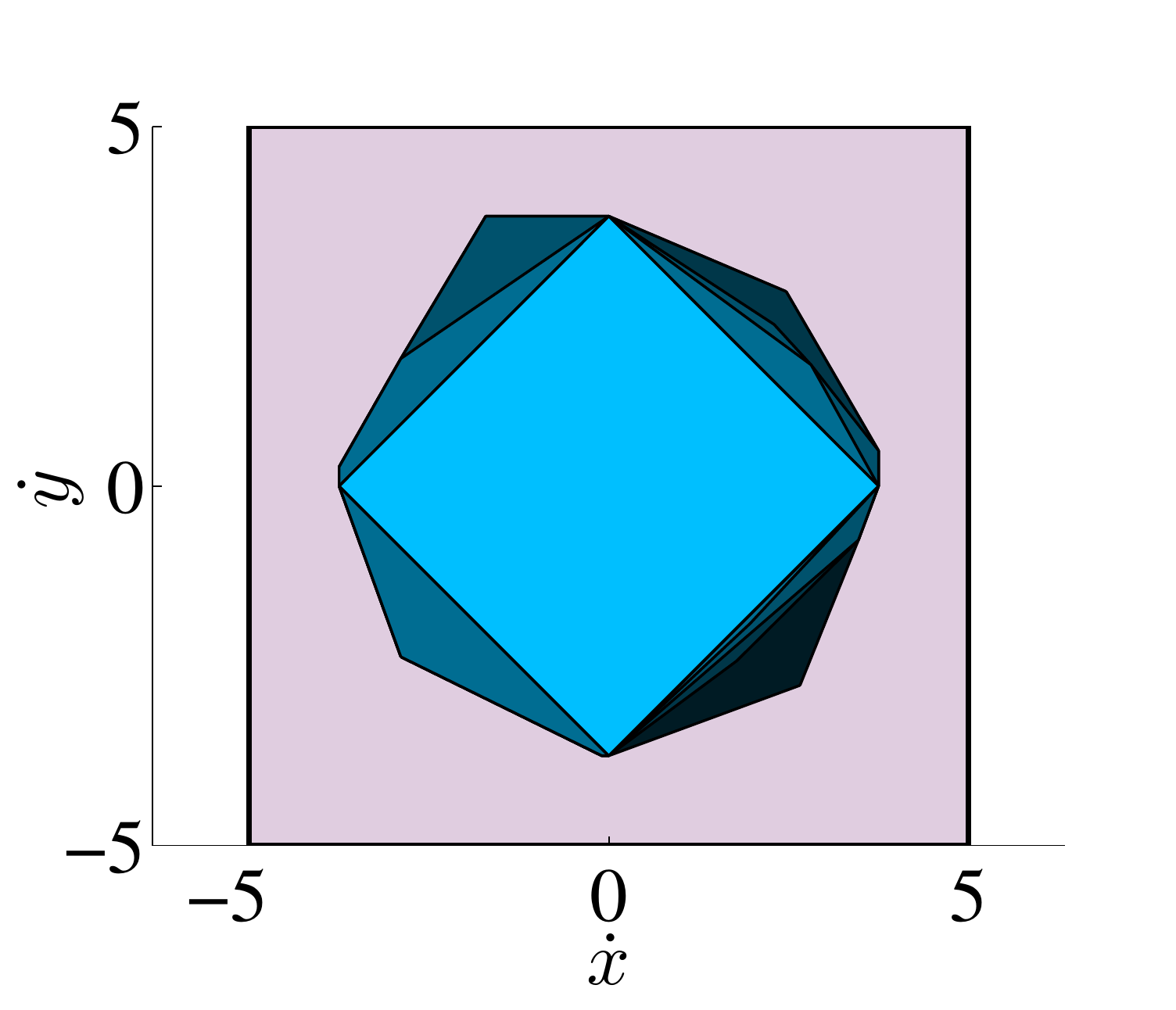}\spcc
    \includegraphics[clip = true, width=\constt\textwidth]{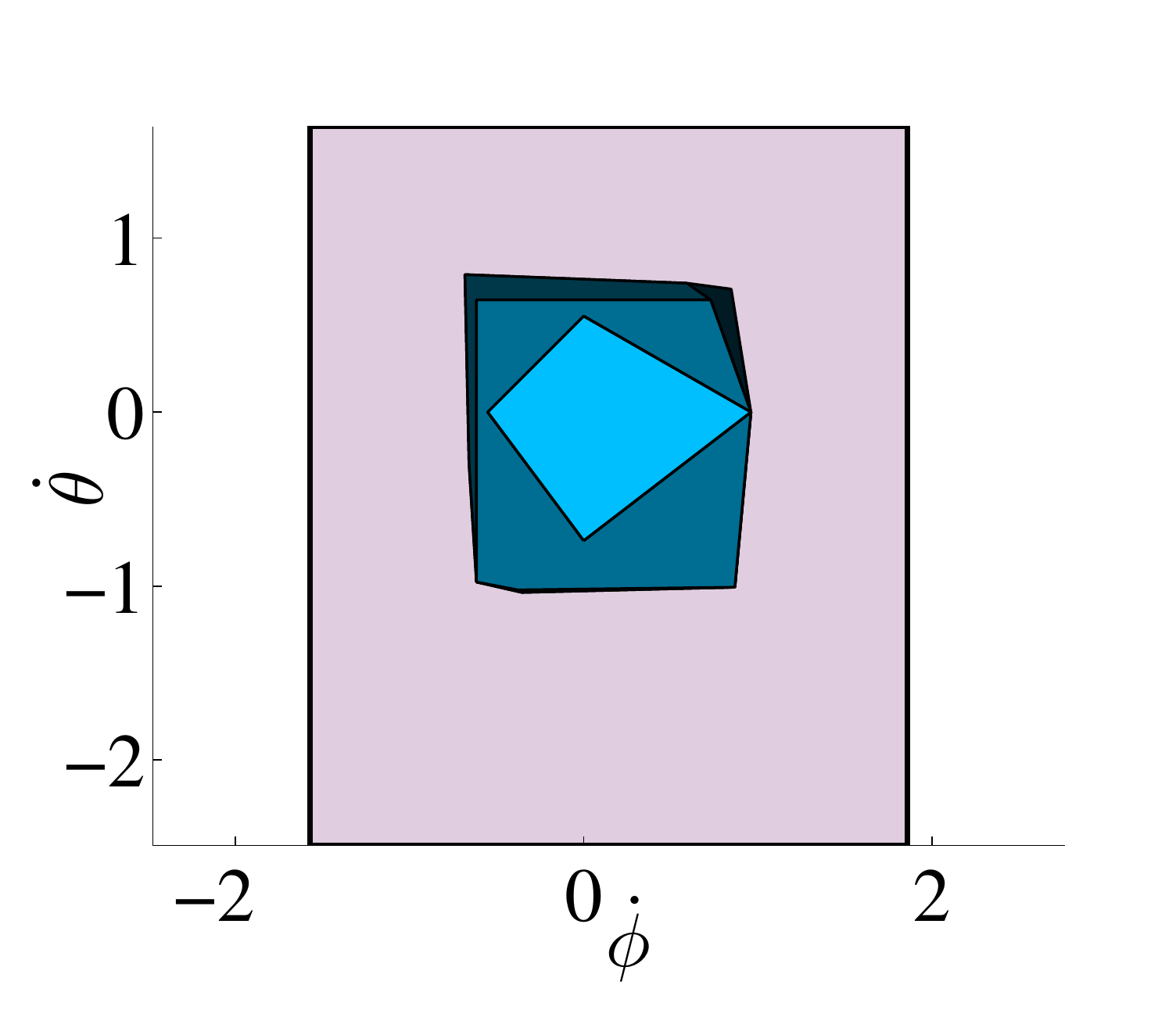}
    \caption{Selected 2D projections of the polytopic approximation of $\Viab_\T(\K)$ for the flight envelope example. Under-approximations for $N=24, 48, 96, 500, 1000, 2000, 3000$ vertices are shown with $N=24$ in the lightest shade of blue (innermost set), and $N=3000$ in the darkest shade of blue. An over-approximation (outermost set) with $N/2$ facets is also shown in lavender. The approximations are tight and touch each other to within a constant accuracy in at least $N/2$ points---a fact that is obscured by the projections.}
    \label{F:quad12D}
\end{figure*}

\section{Conclusions, Extensions, and Future Work}\label{S:conc}

We presented a scalable sampling-based algorithm to generate tight under- and over-approximations of the viability kernel under SD LTI dynamics. We provided correctness and convergence results, discussed a number of heuristics to bias the sampling process for improved performance, and demonstrated the algorithm on a 12D problem of flight envelope protection for an autonomous quadrotor.

The extension to discrete-time systems is straightforward: The set $\K$ is directly used in the feasibility program~\eqref{E:optimization_general} and thus the vertices of the resulting under-approximation and the facets of the over-approximation touch the boundary of the true viability kernel exactly. The algorithm can also be extended to piecewise affine SD (or discrete-time) system by recasting the optimization problem as a mixed-integer program.
Due to space limitations we will present our results on the synthesis of the safety-preserving controllers as well as sufficient conditions under which a generated under-approximation is controlled-invariant (and thus could be used to enforce infinite-horizon safety) in a separate paper.

Robustifying our analysis against unknown but bounded disturbances is challenging (due to the minimax nature of the problem), but a sufficient solution is straightforward and can be done in a number of ways. For instance, the disturbance set could be propagated forward in time, according to which the constraint set could be further eroded to take into account the effect of this uncertainty in future time steps. Alternatively, a pre-stabilization technique ~\cite{Chisci_Rossiter_Zappa_2001} could be used, under certain additional assumptions, such that a portion of the control is dedicated to managing the growth of the disturbance set propagation while the remaining portion is employed to keep the trajectory in $\K$.

The simulation based nature of our algorithm readily admits incorporation of time delays. We will make our initial results on SD systems with transport delays available as a technical report for the interested reader. Finally, we also mention that a similar sampling-based technique can be formulated to approximate nonconvex maximal reachable tubes under LTI dynamics, leveraging the fact that the underlying reachable sets (time slices of the tube) are convex.

Future work for piecewise affine systems includes finding alternative ways to solve the mixed-integer program so as to improve efficiency. The extension of the algorithm to nonconvex constraints remains another challenging problem.




\appendix
%
\section{Appendix: Proof of Proposition~\ref{prop:Convergence}}\label{S:appendix_convergence}

To prove Proposition~\ref{prop:Convergence} we will introduce and make use of a few well-known results from the literature of random algorithms for estimating the boundary of a convex body.

\begin{lem}[\cite{Schutt2003}]
	\label{lem:Convergence_Rate}
	Let $\C$ be a convex body in $\Real^n$ with $\partial \C$ $C^2$, let $f : \partial \C \to \Real_+$ be a probability density function defined on $\partial \C$, let $\mathbb{P}_f$ be the probability measure defined by $f$, and let $\mathbb{E}(f, N)$ be the expected volume of the convex hull of $N$ points chosen randomly on $\partial \C$ with respect to $\mathbb{P}_f$.  Then
	\begin{equation}
		\lim_{N \to \infty} \left(\vol(\C) - \mathbb{E}(f, N)\right) N^\frac{2}{n-1} = c_n(\C),
	\end{equation}
	where $c_n(\C)$ is a constant which depends only on the dimension $n$, the distribution of $f$, and the shape of $\C$.
\end{lem}


To use this result, we also need the following lemma, which introduces a fictitious source of error between the outcome of our algorithm and the viability kernel:
\begin{lem}[\cite{Ghomi2004}]
	\label{lem:Smooth_Approximations_of_Convex_Bodies}
	For every compact convex set $\C$, there exists a compact convex set $\C' \subseteq \C$ whose boundary $\partial \C'$ is $C^2$, and $\vol(\C) - \vol(\C') = \epsilon_{\mathrm{smooth}}$ for some arbitrarily small positive scalar $\epsilon_{\mathrm{smooth}}$.
\end{lem}


Define via Lemma~\ref{lem:Smooth_Approximations_of_Convex_Bodies}, a convex body $\overline{\ViabOp}'$ that is a $C^2$ approximation of $\overline{\Viab_\T}(\K, \zeta, \epsilon)$ such that
\begin{equation}
	\vol(\overline{\Viab_\T}(\K, \zeta, \epsilon)) - \vol(\overline{\ViabOp}') = \epsilon_{\mathrm{smooth}}.
	\label{eqn:Definition_of_Viab_Prime}
\end{equation}

%
\begin{lem}\label{Lem:homeomorphism}
	Any compact convex set $\C$ with $r_0 \in \C$ is homeomorphic to $\B^n_2(r_0,1)$~\cite{Bredon1993}, and in particular we can define the invertible mapping $m : S_2^{n-1}(r_0) \to \partial \C$ as $m(r_d) \equiv \textsc{Bisection-Feasibility}\bigl(r_0,\\ \textsc{Find-Intersection-on-Boundary}(\C, \vec{r}), \C\bigr)$, where $\vec{r}$ has origin $r_0 \in \C$ and direction $r_d$.
\end{lem}

We are now ready to prove Proposition~\ref{prop:Convergence}.

\begin{proof}[Proof of Proposition~\ref{prop:Convergence}]
    Let $f(x)$ be the probability distribution used by $\textsc{Sample-Ray}(v_0)$ to generate samples on the unit sphere.  Then by using the mapping $m$ from Lemma~\ref{Lem:homeomorphism}, we can perform a change of variables to define a new probability density function $g(m(r_d))$ on $\partial \overline{\ViabOp}'$~\cite{Pitman1993}.  Then by Lemma~\ref{lem:Convergence_Rate} we have
    \begin{equation}\label{eqn:Convergence_Proof_Step_1}
    		\lim_{N \to \infty} \left(\vol(\overline{\ViabOp}') - \mathbb{E}(g, N)\right) N^\frac{2}{n-1} = \tilde{c}_n(\overline{\ViabOp}')
    \end{equation}
    for some constant $\tilde{c}_n(\overline{\ViabOp}')$ which depends only on the dimension $n$, the distribution $g$, and the shape of $\overline{\ViabOp}'$.

    Since by Lemma~\ref{lem:Smooth_Approximations_of_Convex_Bodies} $\overline{\ViabOp}'$ can be made arbitrarily close to $\overline{\Viab_\T}(\K,\zeta,\epsilon)$ (i.e. we can take $\epsilon_\mathrm{smooth} \to 0$), the distribution $g$ is mapped almost identically on $\partial \overline{\Viab_\T}(\K,\zeta,\epsilon)$ and we can write \eqref{eqn:Convergence_Proof_Step_1} as
    \begin{multline}\label{eqn:Convergence_Proof_Step_2}
        	\lim_{N \to \infty} \left(\vol(\overline{\Viab_\T}(\K,\zeta,\epsilon)) - \mathbb{E}(g, N) \right) N^\frac{2}{n-1}\\
    		= \lim_{N \to \infty} \left(\vol(\overline{\Viab_\T}(\K,\zeta,\epsilon)) - \vol(\V_N^{\zeta,\epsilon}) \right) N^\frac{2}{n-1}\\
    		= \tilde{c}_n(\overline{\Viab_\T}(\K,\zeta,\epsilon))
    \end{multline}
    since $\mathbb{E}(g, N) = \vol(\V_N^{\zeta, \epsilon})$.

    Taking the limit as $\zeta \to \infty$ and $\epsilon \to 0$ on \eqref{eqn:Convergence_Proof_Step_2} gives
    \begin{equation}
        \lim_{\substack{N \to \infty\\ \zeta \to \infty\\ \epsilon \to 0}} \left(\vol(\overline{\Viab_\T}(\K)) - \vol(\V_N^{\zeta,\epsilon}) \right) N^\frac{2}{n-1}
    		= \tilde{c}_n(\overline{\Viab_\T}(\K)).
    \end{equation}
    Now, by \eqref{eqn:Definition_of_Epsilon_Cont}, we can replace $\vol(\overline{\Viab_\T}(\K))$ to get
    \begin{equation}
        \lim_{\substack{N \to \infty\\ \zeta \to \infty\\ \epsilon \to 0}} \left(\vol(\Viab_\T(\K)) -\epsilon_{\mathrm{cont}}(M\delta) - \vol(\V_N^{\zeta,\epsilon}) \right)
         N^\frac{2}{n-1} = \tilde{c}_n(\overline{\Viab_\T}(\K)).
    \end{equation}

    The shape of $\overline{\Viab_\T}(\K)$ depends only on the shape of $\Viab_\T(\K)$ and the value $M\delta$. Thus there exists an appropriate function $c_n$ such that $\tilde{c}_n(\overline{\Viab_\T}(\K)) = c_n(\Viab_\T(\K),M\delta)$. Consequently, we arrive at \eqref{E:convergence_rate}.
    %
\end{proof}

\bibliographystyle{IEEEtran}
\bibliography{IEEEabrv,viabET_v2,library,sos_based}            

\begin{thebibliography}{10}
\providecommand{\url}[1]{#1}
\csname url@rmstyle\endcsname
\providecommand{\newblock}{\relax}
\providecommand{\bibinfo}[2]{#2}
\providecommand\BIBentrySTDinterwordspacing{\spaceskip=0pt\relax}
\providecommand\BIBentryALTinterwordstretchfactor{4}
\providecommand\BIBentryALTinterwordspacing{\spaceskip=\fontdimen2\font plus
\BIBentryALTinterwordstretchfactor\fontdimen3\font minus
  \fontdimen4\font\relax}
\providecommand\BIBforeignlanguage[2]{{%
\expandafter\ifx\csname l@#1\endcsname\relax
\typeout{** WARNING: IEEEtran.bst: No hyphenation pattern has been}%
\typeout{** loaded for the language `#1'. Using the pattern for}%
\typeout{** the default language instead.}%
\else
\language=\csname l@#1\endcsname
\fi
#2}}

\bibitem{Gillula_HSCC2014}
J.~H. Gillula, S.~Kaynama, and C.~J. Tomlin, ``Sampling-based approximation of
  the viability kernel for high-dimensional linear sampled-data systems,'' in
  \emph{Hybrid Systems: Computation and Control}, Berlin, Germany, 2014, pp.
  173--182.

\bibitem{aubin2011}
J.-P. Aubin, A.~M. Bayen, and P.~Saint-Pierre, \emph{Viability Theory: New
  Directions}, 2nd~ed.\hskip 1em plus 0.5em minus 0.4em\relax Springer Verlag,
  2011.

\bibitem{blanchini2008set}
F.~Blanchini and S.~Miani, \emph{Set-Theoretic Methods in Control}.\hskip 1em
  plus 0.5em minus 0.4em\relax Springer, 2008.

\bibitem{Goodwin_CSSMAG13}
G.~Goodwin, J.~Aguero, M.~Cea~Garridos, M.~Salgado, and J.~Yuz, ``Sampling and
  sampled-data models: The interface between the continuous world and digital
  algorithms,'' \emph{IEEE Control Systems Magazine}, vol.~33, no.~5, pp.
  34--53, 2013.

\bibitem{diabetes_MPC_2007}
L.~Magni, D.~M. Raimondo, L.~Bossi, C.~D. Man, G.~D. Nicolao, B.~Kovatchev, and
  C.~Cobelli, ``Model predictive control of type 1 diabetes: An in silico
  trial,'' \emph{Journal of Diabetes Science and Technology}, vol.~1, no.~6,
  pp. 804--812, 2007.

\bibitem{MBT05}
I.~M. Mitchell, A.~M. Bayen, and C.~J. Tomlin, ``A time-dependent
  {Hamilton-Jacobi} formulation of reachable sets for continuous dynamic
  games,'' \emph{IEEE Transactions on Automatic Control}, vol.~50, no.~7, pp.
  947--957, July 2005.

\bibitem{kaynama_NAHS2012}
I.~M. Mitchell, S.~Kaynama, M.~Chen, and M.~Oishi, ``Safety preserving control
  synthesis for sampled data systems,'' \emph{Nonlinear Analysis: Hybrid
  Systems}, vol.~10, pp. 63--82, 2013.

\bibitem{Saint-Pierre_1994}
P.~Saint-Pierre, ``Approximation of the viability kernel,'' \emph{Applied
  Mathematics and Optimization}, vol.~29, no.~2, pp. 187--209, Mar 1994.

\bibitem{kaynama_TAC2009}
S.~Kaynama and M.~Oishi, ``A modified {Riccati} transformation for
  decentralized computation of the viability kernel under {LTI} dynamics,''
  \emph{IEEE Transactions on Automatic Control}, vol.~58, no.~11, pp.
  2878--2892, 2013.

\bibitem{Mitchell_HSCC11}
I.~M. Mitchell, ``Scalable calculation of reach sets and tubes for nonlinear
  systems with terminal integrators: A mixed implicit explicit formulation,''
  in \emph{Hybrid Systems: Computation and Control}, Chicago, IL, 2011, pp.
  103--112.

\bibitem{Coquelin_Martin_Munos_2007}
P.-A. Coquelin, S.~Martin, and R.~Munos, ``A dynamic programming approach to
  viability problems,'' in \emph{IEEE Symposium on Adaptive Dynamic Programming
  and Reinforcement Learning}, 2007, pp. 178--184.

\bibitem{Bremner_1998}
D.~D. Bremner, ``On the complexity of vertex and facet enumeration for convex
  polytopes,'' Ph.D. dissertation, Montreal, QC, Canada, 1998.

\bibitem{chen2001optimisation}
W.-H. Chen, D.~J. Ballance, and J.~O'Reilly, ``Optimisation of attraction
  domains of nonlinear {MPC} via {LMI} methods,'' in \emph{American Control
  Conference}, Arlington, VA, 2001, pp. 3067--3072.

\bibitem{Alessio_Lazar_Bemporad_Heemels_2007}
A.~Alessio, M.~Lazar, A.~Bemporad, and W.~Heemels, ``Squaring the circle: An
  algorithm for generating polyhedral invariant sets from ellipsoidal ones,''
  \emph{Automatica}, vol.~43, no.~12, pp. 2096--2103, 2007.

\bibitem{kaynama_Aut2013}
J.~Maidens, S.~Kaynama, I.~M. Mitchell, M.~Oishi, and G.~A. Dumont,
  ``Lagrangian methods for approximating the viability kernel in
  high-dimensional systems,'' \emph{Automatica}, vol.~49, no.~7, pp.
  2017--2029, 2013.

\bibitem{kaynama_HSCC2012}
S.~Kaynama, J.~Maidens, M.~Oishi, I.~M. Mitchell, and G.~A. Dumont, ``Computing
  the viability kernel using maximal reachable sets,'' in \emph{Hybrid Systems:
  Computation and Control}, Beijing, 2012, pp. 55--63.

\bibitem{Le_Guernic_Girard_2010}
C.~Le~Guernic and A.~Girard, ``Reachability analysis of linear systems using
  support functions,'' \emph{Nonlinear Analysis: Hybrid Systems}, vol.~4,
  no.~2, pp. 250--262, 2010.

\bibitem{SpaceEx_2011}
G.~Frehse, C.~Le~Guernic, A.~Donze, S.~Cotton, R.~Ray, O.~Lebeltel, R.~Ripado,
  A.~Girard, T.~Dang, and O.~Maler, ``{SpaceEx}: Scalable verification of
  hybrid systems,'' in \emph{23rd International Conference on Computer Aided
  Verification}, G.~Gopalakrishnan and S.~Qadeer, Eds.\hskip 1em plus 0.5em
  minus 0.4em\relax Springer, 2011, pp. 1--16.

\bibitem{Kurzhanski_SCL_2000}
A.~B. Kurzhanski and P.~Varaiya, ``Ellipsoidal techniques for reachability
  analysis: internal approximation,'' \emph{Systems Control Letters}, vol.~41,
  pp. 201--211, 2000.

\bibitem{tedrake_finitetime_13}
A.~Majumdar and R.~Tedrake, ``Robust online motion planning with regions of
  finite time invariance,'' in \emph{Algorithmic Foundations of Robotics},
  E.~Frazzoli, T.~Lozano-Perez, N.~Roy, and D.~Rus, Eds.\hskip 1em plus 0.5em
  minus 0.4em\relax Springer Berlin Heidelberg, 2013, pp. 543--558.

\bibitem{tedrake_funnels_11}
M.~Tobenkin, I.~Manchester, and R.~Tedrake, ``Invariant funnels around
  trajectories using sum-of-squares programming,'' in \emph{IFAC World
  Congress}, vol.~18, Milano, Italy, 2011, pp. 9218--9223.

\bibitem{tedrake2010lqr}
R.~Tedrake, I.~R. Manchester, M.~Tobenkin, and J.~W. Roberts, ``{LQR}-trees:
  Feedback motion planning via sums-of-squares verification,''
  \emph{International Journal of Robotics Research}, vol.~29, no.~8, pp.
  1038--1052, 2010.

\bibitem{tan_CLF_04}
W.~Tan and A.~Packard, ``Searching for control lyapunov functions using sums of
  squares programming,'' in \emph{Annual Allerton Conference}, 2004, pp.
  210--219.

\bibitem{Prajna_Jadbabaie_2004}
S.~Prajna and A.~Jadbabaie, ``Safety verification of hybrid systems using
  barrier certificates,'' in \emph{Hybrid Systems: Computation and Control},
  2004, pp. 477--492.

\bibitem{henrion2012convex}
D.~Henrion and M.~Korda, ``Convex computation of the region of attraction of
  polynomial control systems,'' \emph{preprint arXiv:1208.1751}, 2012.

\bibitem{majumdar2013technical}
A.~Majumdar, R.~Vasudevan, M.~M. Tobenkin, and R.~Tedrake, ``Technical report:
  Convex optimization of nonlinear feedback controllers via occupation
  measures,'' \emph{preprint arXiv:1305.7484}, 2013.

\bibitem{Moler_Van_Loan_2003}
C.~Moler and C.~Van~Loan, ``Nineteen dubious ways to compute the exponential of
  a matrix, twenty-five years later,'' \emph{SIAM Review}, vol.~45, no.~1, pp.
  3--49, 2003.

\bibitem{Higham2005MatlabExpm}
N.~J. Higham, ``The scaling and squaring method for the matrix exponential
  revisited,'' \emph{SIAM Journal of Matrix Analysis and Applications},
  vol.~26, no.~4, pp. 1179--1193, 2005.

\bibitem{Standish1975179}
C.~Standish, ``Truncated {Taylor} series approximation to the state transition
  matrix of a continuous parameter finite {Markov} chain,'' \emph{Linear
  Algebra and its Applications}, vol.~12, no.~2, pp. 179--183, 1975.

\bibitem{Liou_1966}
M.~Liou, ``A novel method of evaluating transient response,'' \emph{Proceedings
  of the {IEEE}}, vol.~54, no.~1, pp. 20--23, 1966.

\bibitem{Barany1987}
I.~B\'{a}r\'{a}ny and Z.~F\"{u}redi, ``Computing the volume is difficult,''
  \emph{Discrete Computational Geometry}, vol.~2, no.~1, pp. 319--326, 1987.

\bibitem{Dyer1991}
M.~Dyer, ``{Computing the volume of convex bodies: a case where randomness
  provably helps},'' \emph{Probabilistic Combinatorics and its Applications},
  vol.~44, pp. 123--170, 1991.

\bibitem{Schutt2003}
C.~Sch\"{u}tt and E.~Werner, ``Polytopes with vertices chosen randomly from the
  boundary of a convex body,'' in \emph{Geometric Aspects of Functional
  Analysis}, ser. LNM 1807, V.~Milman and G.~Schechtman, Eds.\hskip 1em plus
  0.5em minus 0.4em\relax Springer Berlin, 2003, pp. 241--422.

\bibitem{Yalmip_Lofberg04}
J.~L\"{o}fberg, ``{YALMIP}: a toolbox for modeling and optimization in
  {MATLAB},'' in \emph{Computer Aided Control System Design}, 2004, pp.
  284--289.

\bibitem{KGBM04}
M.~Kvasnica, P.~Grieder, M.~Baoti\'{c}, and M.~Morari, ``{Multi-Parametric
  Toolbox (MPT)},'' in \emph{Hybrid Systems: Computation and Control, LNCS
  2993}, R.~Alur and G.~J. Pappas, Eds.\hskip 1em plus 0.5em minus 0.4em\relax
  Berlin, Germany: Springer, 2004, pp. 448--462.

\bibitem{Mardia_directionalstat}
K.~V. Mardia and P.~E. Jupp, \emph{Directional Statistics}, 2nd~ed.\hskip 1em
  plus 0.5em minus 0.4em\relax John Wiley \& Sons, Inc, 2000.

\bibitem{Dyer_vol_NPhard88}
M.~E. Dyer and A.~M. Frieze, ``On the complexity of computing the volume of a
  polyhedron,'' \emph{SIAM Journal of Computing}, vol.~17, no.~5, pp. 967--974,
  1988.

\bibitem{Boyd_cvxbook}
S.~P. Boyd and L.~Vandenberghe, \emph{Convex optimization}.\hskip 1em plus
  0.5em minus 0.4em\relax Cambridge University Press, 2004.

\bibitem{kaynama_benchmark}
S.~Kaynama and C.~J. Tomlin, ``Benchmark: Flight envelope protection in
  autonomous quadrotors,'' in \emph{International Workshop on Applied
  Verification for Continuous and Hybrid Systems, Part of CPSWeek}, Berlin,
  Germany, April 2014, http://cps-vo.org/group/ARCH.

\bibitem{Cowling_quadrotor_paper}
I.~Cowling, O.~Yakimenko, J.~Whidborne, and A.~Cooke, ``Direct method based
  control system for an autonomous quadrotor,'' \emph{Journal of Intelligent \&
  Robotic Systems}, vol.~60, pp. 285--316, 2010.

\bibitem{Chisci_Rossiter_Zappa_2001}
L.~Chisci, J.~Rossiter, and G.~Zappa, ``Systems with persistent disturbances:
  predictive control with restricted constraints,'' \emph{Automatica}, vol.~37,
  pp. 1019--1028, 2001.

\bibitem{Ghomi2004}
M.~Ghomi, ``{Optimal Smoothing for Convex Polytopes},'' \emph{Bulletin of the
  London Mathematical Society}, vol.~36, no.~4, pp. 483--492, July 2004.

\bibitem{Bredon1993}
G.~E. Bredon, \emph{{Topology and Geometry}}.\hskip 1em plus 0.5em minus
  0.4em\relax Springer, 1993, vol. 139.

\bibitem{Pitman1993}
J.~Pitman, \emph{{Probability}}.\hskip 1em plus 0.5em minus 0.4em\relax New
  York: Springer-Verlag, 1993.

\end{thebibliography}

\end{document}